\documentclass[a4paper,11pt]{article}

\usepackage[a4paper, margin={1in, 1in}]{geometry}

\usepackage{amssymb,amsfonts,amsmath}
\usepackage{amsbsy}
\usepackage{amsthm}
\usepackage{hyperref}
\usepackage[linesnumbered,ruled]{algorithm2e}
\usepackage{footnote}
\usepackage{cleveref,aliascnt}

\usepackage{arydshln}

\usepackage{multirow}
\usepackage{hhline}
\usepackage{makecell} 

\newcommand{\N}{\mathbb{N}}                     
\newcommand{\R}{\mathbb{R}}                     

\newcommand{\Z}{\mathbb{Z}}  			  	 
               
\newcommand{\poly}{\text{poly}}
	
\newcommand{\eps}{\varepsilon}

\newcommand{\wh}[1]{\widehat{#1}}

\newcommand{\E}{\mathbb{E}}

\newcommand{\pr}[1]{\text{\bf Pr}\normalfont\lbrack #1 \rbrack} 
\newcommand{\ex}[1]{\mathbb{E}\normalfont\lbrack #1 \rbrack}
\newcommand{\bex}[1]{\mathbb{E}\normalfont \Big[#1 \Big]}

\newtheorem{theorem}{Theorem}[section]

\newtheorem{lemma}[theorem]{Lemma}

\newtheorem{corollary}[theorem]{Corollary}

\newtheorem{proposition}[theorem]{Proposition}

\theoremstyle{definition}
\newtheorem{definition}[theorem]{Definition}

\newtheorem{remark}[theorem]{Remark}

\begin{document}

\title{A Framework for Adversarially Robust Streaming Algorithms}
\author{
	Omri Ben-Eliezer\thanks{Massachusetts Institute of Technology. Work partially conducted while the author was at Tel Aviv University and later at Harvard University. Email: \texttt{omrib@mit.edu}} 
	\and
Rajesh Jayaram\thanks{Google Research. Work partially conducted while the author was at Carnegie Mellon University, where he was supported by the Office of Naval Research (ONR) grant N00014-18-1-2562, and the National Science Foundation (NSF) under Grant No. CCF-1815840. Email: \texttt{rkjayaram@google.com}} 
\and
David P.~Woodruff\thanks{Carnegie Mellon University. Supported by the Office of Naval Research (ONR) grant N00014-18-1-2562, and the National Science Foundation (NSF) under Grant No. CCF-1815840. Email: \texttt{dwoodruf@cs.cmu.edu}}
\and
Eylon Yogev\thanks{Bar-Ilan University. Work partially conducted while the author was at Tel Aviv University. Email: \texttt{eylon.yogev@biu.ac.il}}}
\date{}
\maketitle

\begin{abstract}
We investigate the adversarial robustness of streaming algorithms. In this context, an algorithm is considered robust if its performance guarantees hold even if the stream is chosen adaptively by an adversary that observes the outputs of the algorithm along the stream and can react in an online manner. While deterministic streaming algorithms are inherently robust, many central problems in the streaming literature do not admit sublinear-space deterministic algorithms; on the other hand, classical space-efficient randomized algorithms for these problems are generally not adversarially robust. This raises the natural question of whether there exist efficient adversarially robust (randomized) streaming algorithms for these problems.

In this work, we show that the answer is positive for various important streaming problems in the insertion-only model, including distinct elements and more generally $F_p$-estimation, $F_p$-heavy hitters, entropy estimation, and others.
For all of these problems, we develop adversarially robust $(1+\varepsilon)$-approximation algorithms whose required space matches that of the best known non-robust algorithms up to a $\poly(\log n, 1/\varepsilon)$ 
multiplicative factor (and in some cases even up to a constant factor). Towards this end, we develop several generic tools allowing one to efficiently transform a non-robust streaming algorithm into a robust one in various scenarios.
\end{abstract}

\section{Introduction}
The streaming model of computation is a central and crucial tool for the analysis of massive datasets, where the sheer size of the input imposes stringent restrictions on the memory, computation time, and other resources available to the algorithms. Examples of theoretical and practical settings where streaming algorithms are in need are easy to encounter. These include internet routers and traffic logs, databases, sensor networks, financial transaction data, and scientific data streams. Given this wide range of applicability, there has been significant effort devoted to designing and analyzing extremely efficient one-pass algorithms. We recommend the survey of \cite{Muthu} for a comprehensive overview of streaming algorithms and their applications.

Many central problems in the streaming literature do not admit sublinear-space deterministic algorithms, and in these cases randomized solutions are necessary. In other cases, randomized solutions are more efficient and simpler to implement than their deterministic counterparts. While randomized streaming algorithms are well-studied, the vast majority of them are defined and analyzed in the {\em static} setting, where the stream is worst-case but fixed in advance, and only then the randomness of the algorithm is chosen. However, assuming that the stream sequence is independent of the chosen randomness, and in particular that future elements of the stream do not depend on previous outputs of the streaming algorithm, may not be realistic \cite{mironov2011sketching,gilbert2012recovering,gilbert2012reusable,HardtW13,NaorY15,BenEliezerY19,ABDMNY21}, even in non-adversarial settings. For example, suppose that a user sequentially makes updates in a database, and receives an immediate response about the current state of the data after each update. Naturally, future updates made by the user in such a setting may heavily depend on the responses given by the database to previous queries. In other words, the stream updates are chosen adaptively, and cannot be assumed to be fixed in advance.

A streaming algorithm that works even when the stream is adaptively chosen by an adversary (the precise definition given next) is said to be {\em adversarially robust}. Deterministic algorithms are inherently adversarially robust, since they are guaranteed to be correct on all possible inputs. However, the large gap in performance between deterministic and randomized streaming algorithms for many problems motivates the need for designing adversarially robust randomized algorithms, if they even exist. In particular, we would like to design adversarially robust randomized algorithms which are as space and time efficient as their static counterparts, and yet as robust as deterministic algorithms. The study of such algorithms is the main focus of our work.

\paragraph{The Adversarial Setting.}
There are several ways to define the adversarial setting, which may depend on the information the adversary (who chooses the stream) can observe from the streaming algorithm, as well as other restrictions imposed on the adversary. 
For the most part, we consider a general model, where the adversary is allowed unbounded computational power and resources, though we do discuss the case later when the adversary is computationally bounded. At each point in time, the streaming algorithm publishes its output to a query for the stream. 
The adversary observes these outputs one-by-one, and can choose the next update to the stream adaptively, depending on the full history of the outputs and stream updates. The goal of the adversary is to force the streaming algorithm to eventually produce an \textit{incorrect} output to the query, as defined by the specific streaming problem in question.\footnote{In the streaming literature, an algorithm is often required to be correct on a query made only \textit{once}, at the end of the stream. This is a \textit{one-shot} guarantee, as opposed to the \textit{tracking} 
	guarantee as defined here. However, the two settings are nearly equivalent. Indeed, for almost all streaming problems, a one-shot algorithm can be made into a tracking algorithm with at most an $O(\log n)$ blow-up in space, by simply setting the failure probability small enough to union bound over all points in the stream.}

Formally, a data stream of length $m$ over a domain $[n]$ is a sequence of updates of the form $(a_1,\Delta_1),(a_2,\Delta_2),\dots,(a_m,\Delta_m)$ where $a_t \in [n]$ is an index and $\Delta_t \in \Z$ is an increment or decrement to that index. The \textit{frequency vector} $f \in \R^n$ of the stream is the vector with $i^{\text{th}}$ coordinate $f_i = \sum_{t: a_t = i} \Delta_t$. We write $f^{(t)}$ to denote the frequency vector restricted to the first $t$ updates, namely $f_i^{(t)} = \sum_{j \leq t : a_j = i} \Delta_j$. It is assumed at all points $t$ that the maximum coordinate in absolute value, denoted $\|f^{(t)}\|_\infty$, is at most $M$ for some $M > 0$, and that $\log(mM) = O( \log n)$.  In the \textit{insertion-only} model, the updates are assumed to be positive, meaning $\Delta_t > 0$, whereas in the \textit{turnstile} model $\Delta_t$ can be positive or negative.

The general task in streaming is to respond to some query $\mathcal{Q}$ about the frequency vector $f^{(t)}$ at each point in time $t \in [m]$. Oftentimes, this query is to approximate\footnote{Ideally, one might wish to exactly compute the function $g$; however, in many cases, and in particular for the problems that we consider here, exact computation cannot be done with sublinear space.} some function $g:\R^n \to \R$ of $f^{(t)}$. For example, counting the number of distinct elements in a data stream is among the most fundamental problems in the streaming literature; here $g(f^{(t)})$ is the number of non-zero entries in $f^{(t)}$. Since exact computation cannot be done in sublinear space \cite{chakrabarti2016strong}, the goal is to approximate the value of $g(f^{(t)})$ to within a multiplicative factor of $(1 \pm \eps)$. Another important streaming problem (which is not directly an estimation task) is the \textit{Heavy-Hitters} problem, where the algorithm is tasked with finding all the coordinates in $f^{(t)}$ which are larger than some threshold $\tau$.


Formally, the adversarial setting is modeled by a two-player game between a (randomized) \textsc{StreamingAlgorithm} and an \textsc{Adversary}. At the beginning, a query $\mathcal{Q}$ is fixed, which the \textsc{StreamingAlgorithm}
must continually reply to. The game proceeds in rounds, where in the $t$-th round:
\begin{enumerate}
	\item \textsc{Adversary} chooses an update $u_t = (a_t ,\Delta_t)$ for the stream, which can depend, in particular, on all previous stream updates and outputs of \textsc{StreamingAlgorithm}.
	\item \textsc{StreamingAlgorithm} processes the new update $u_t$ and outputs its current response $R^t$ to the query $\mathcal{Q}$.
	\item \textsc{Adversary} observes $R^t$ (stores it) and proceeds to the next round.
\end{enumerate}
The goal of the \textsc{Adversary} is to make the \textsc{StreamingAlgorithm} output an incorrect response $R^t$ to $\mathcal{Q}$ at some point $t$ in the stream. For example, in the distinct elements problem, the adversary's goal is that at some step $t$, the estimate $R^t$ will fail to be a $(1+\eps)$-approximation of the true current number of distinct elements $|\{i \in [n] : f^{(t)}_i \neq 0 \}|$.

\paragraph{Streaming algorithms in the adversarial setting.}
It was shown by Hardt and Woodruff \cite{HardtW13} that linear sketches are inherently \emph{non-robust} in adversarial settings for a large family of problems, thus demonstrating a major limitation of such sketches. In particular, their results imply that no linear sketch can approximate the Euclidean norm of its input to within a polynomial multiplicative factor in the adversarial (turnstile) setting. Here, a linear sketch is an algorithm whose output depends only on values $A f$ and $A$, for some (usually randomized) sketching matrix $A \in \R^{k \times n}$. This is quite unfortunate, as the vast majority of turnstile streaming algorithms are in fact linear sketches.

On the positive side, a recent work of Ben-Eliezer and Yogev \cite{BenEliezerY19} (see also \cite{ABDMNY21, braverman2021adversarial}) showed that \textit{random sampling} is quite robust in the adaptive adversarial setting, albeit with a slightly larger sample size. While uniform sampling is a rather generic and important tool, it is not sufficient for solving many important streaming tasks, such as estimating frequency moments ($F_p$-estimation), finding $L_2$ heavy hitters, and various other central data analysis problems. This raises the natural question of whether there exist efficient adversarially robust randomized streaming algorithms for these problems and others, which is the main focus of this work. Perhaps even more importantly, we ask the following. 
\begin{quote}
	\begin{center}
		{\em
			Is there a generic technique to transform a static streaming algorithm\\into an adversarially robust streaming algorithm?
		}
	\end{center}
\end{quote}
This work answers the above questions affirmatively for a large class of algorithms.

\subsection{Our Results}
We devise adversarially robust algorithms for various fundamental insertion-only streaming problems, including distinct element estimation, $F_p$ moment estimation, heavy hitters, entropy estimation, and several others. In addition, we give adversarially robust streaming algorithms which can handle a bounded number of deletions as well. The required space of our adversarially robust algorithms matches that of the best known non-robust ones up to a small multiplicative factor. Our new algorithmic results are summarized in Table \ref{table:summary}.
In contrast, we demonstrate that some classical randomized algorithms for streaming problems in the static setting, such as the celebrated Alon-Matias-Szegedy (AMS) sketch \cite{AlonMS96} for $F_2$-estimation, are inherently non-robust to adaptive adversarial attacks in a strong sense, even against an insertion-only adaptive adversary. In comparison, the attack of Hardt and Woodruff on linear sketches \cite{HardtW13} requires both insertions and deletions.

Our adversarially robust algorithms make use of two generic robustification frameworks that we develop, allowing one to efficiently transform a non-robust streaming algorithm into a robust one in various settings. Both of the robustification methods rely on the fact that functions of interest do not drastically change their value too many times along the stream. Specifically, the transformed algorithms have space dependency on the \textit{flip-number} of the stream, which is a bound on the number of times the function $g(f^{(t)})$ can change by a factor of $(1\pm\eps)$ in the stream (see Section \ref{sec:framework}).

The first method, called \emph{sketch switching}, maintains multiple instances of the non-robust algorithm and switches between them in a way that cannot be exploited by the adversary. The second technique bounds the number of \emph{computation paths} possible in the two-player adversarial game. This technique maintains only one copy of a non-robust algorithm, albeit with an extremely small probability of error $\delta$. We show that a carefully rounded sequence of outputs generates only a small number of possible computation paths, which can then be used to ensure robustness by union bounding over these paths. The framework is described in Section \ref{sec:framework}.

The two above methods are incomparable: for some streaming problems the former is more efficient, while for others, the latter performs better, and we show examples of each. 
Specifically, sketch switching can exploit efficiency gains of \textit{strong-tracking}, resulting in particularly good performance for static algorithms that can respond correctly to queries at each step without having to union bound over all $m$ steps. In contrast, the computation paths technique can exploit an algorithm with good dependency on $\delta$ (the failure probability). Namely, algorithms that have small dependency in update-time or space on $\delta$ will benefit from the computation paths technique.

\begin{table}
\begin{center}
\begin{tabular}{|l|l|l|l|l|}
\hhline{-----}
\textbf{Problem}	& \textbf{Static Rand.} & \textbf{Deter.}\hspace{-1pt}		& \textbf{Adversarial} & \textbf{Comments}  \\ 
\hhline{=====}

Distinct elem. 
& \multirow{2}{*}{$\tilde{O}(\eps^{-2} + \log n)$
}
& \multirow{2}{*}{$\Omega(n)$\hspace{-1pt}
} 
& \small $\boldsymbol{\tilde{O}(\eps^{-3} + \eps^{-1} \log n)}$
&\\
\hhline{~~~--}
($F_0$ est.)
&&
& \small $\boldsymbol{\tilde{O}(\eps^{-2} + \log n)}$
& \small crypto/rand. oracle
\\ \hhline{-----}

$F_p$ estimation, 
& $O(\eps^{-2}\log n)$
& \multirow{2}{*}{$\tilde{\Omega}(c_p n)$\hspace{-1pt}}
& \small $\boldsymbol{\tilde{O}(\eps^{-3} \log n)}$ & 
\\ \hhline{~-~--}
$p \in (0,2] \setminus \{1\}$

 & $O(\eps^{-3}\log^2 n)$  
 & 
 & \small $\boldsymbol{\tilde{O}(\eps^{-3} \log^3 n)}$ & 
 \small $\delta = \Theta( n^{-\frac{1}{\eps} \log n} )$
\\  \hhline{-----}

$F_p$ estimation, &
\small $O(n^{1-\frac{2}{p}}(\eps^{-2}\log n$ \hspace{-20pt}&
\multirow{2}{*}{$\Omega(n)$\  
} &
\small $\boldsymbol{O(n^{1-\frac{2}{p}}(\eps^{-3}\log^2 n}$ &
\small \multirow{2}{*}{$\delta = \Theta( n^{-\frac{1}{\eps} \log n} )$} 
\\ \hhline{~~~~~}
$p>2$
& 
\small $+\eps^{-\frac{4}{p}} \log^{\frac{2}{p}+1} n))$ 
&
&
\small $\boldsymbol{+\eps^{-\frac{6}{p}} \log^{\frac{4}{p}+1} n))}$ &
\\ \hhline{-----}

{$\ell_2$ Heavy Hit.}  & 
$O(\eps^{-2} \log^2 n)$
& $\Omega(\sqrt{n})$
& \small $\boldsymbol{\tilde{O}(\eps^{-3} \log^2 n)}$ & 
\\ \hhline{-----}

Entropy &  
$O(\eps^{-2} \log^3 n)$  
& \multirow{2}{*}{$\tilde{\Omega}(n)$} & \small $\boldsymbol{\tilde{O}(\eps^{-4} \log^{6}n)}$ & \\
\hhline{~-~--}
estimation
&
$\tilde{O}(\eps^{-2} \log n)$ 
&
& \small $\boldsymbol{\tilde{O}(\eps^{-4} \log^{4}n)}$     
& \small crypto/rand. oracle
\\ \hhline{-----}

Turnstile $F_p$, 
& \multirow{2}{*}{$O(\eps^{-2}  \log^2 n)$\  
} & \multirow{2}{*}{$\Omega(n )$  
} & \small \multirow{2}{*}{$\boldsymbol{O(\eps^{-2} \lambda \log^2 n)}$} & 
\small $\lambda$-bounded $F_p$ flip \\
\hhline{~~~~~}
$p \in (0,2]$
&&&&
    \small num., $\delta = \Theta(n^{-\lambda})$
\\  \hhline{-----}

$F_p$, $p \in [1,2]$ &
$\tilde{O}(\log^2 n +$ 
& \multirow{2}{*}{$\tilde{\Omega}(c_p n)$} 
& \small \multirow{2}{*}{$\boldsymbol{O(\alpha \eps^{-(2+p)} \log^3 n)}$} & \small static only 
\\ \hhline{~~~~~}
$\alpha$-bounded del.
&
$\eps^{-2} \log \alpha \log n)$
&
\small 
&& \small for $p=1$
\\ \hhline{-----}
\end{tabular}
\end{center}
\caption{A summary of our adversarially robust algorithms (in bold), 
as compared to the best known upper bounds for randomized algorithms 
in the static setting and lower bounds for deterministic algorithms. The space bounds are given in bits.
Note that all stated algorithms provide tracking. 
All results except for the last two (which hold in 
restricted versions of the turnstile model) are for insertion-only streams.  
We write $\tilde{O},\tilde{\Omega}$ to hide $\log \eps^{-1}$ and $\log \log n$ factors.
The static randomized upper bounds are proved, respectively, in \cite{blasiok2018optimal}, \cite{blasiok2017continuous}, \cite{kane2010exact}, \cite{ganguly2018high}, \cite{braverman2017bptree}, \cite{clifford2013simple}, \cite{jayaram2019towards}, \cite{kane2010exact}, and \cite{jayaram2018data}. All lower bounds for $F_p$-estimation are proved in \cite{chakrabarti2016strong}, except for the turnstile bound, proved in \cite{AlonMS96}; the lower bound for heavy hitters is from \cite{kpw20}. Finally,
the lower bound for deterministic entropy estimation follows from a 
reduction from estimating $F_p$ for $p = 1+ \tilde{\Theta}(\eps/\log^2 n)$ to 
entropy estimation \cite{harvey2008sketching}.}
\label{table:summary}
\end{table}

For each of the problems we consider, we show how to use the framework combined with some additional techniques, to solve it.
Interestingly, we also demonstrate how cryptographic assumptions (which were not commonly used before in the streaming context) can be applied to obtain an adversarially robust algorithm against computationally bounded adversaries for the distinct elements problem at essentially no extra cost (compared to the space-optimal non-robust algorithm). See Table \ref{table:summary} for a summary of our results in the adversarial setting compared to the state-of-the-art in the static setting, as well as to deterministic algorithms.

\paragraph{Distinct elements and $F_p$-estimation}
Our first suite of results provides robust streaming algorithms for estimating $F_p$, the $p^{\text{th}}$ frequency moment of the frequency vector, defined as $F_p = \|f\|_p^p = \sum_{i=1}^n |f_i|^p$, where we interpret $0^0 = 0$. 
Estimating frequency moments has a myriad of applications in databases, computer networks, data mining, and other contexts. Efficient algorithms for estimating distinct elements (i.e., estimating $F_0$) are important for databases, since query optimizers can use them to find the number of unique values of an attribute without having to perform an expensive sort on the values. Efficient algorithms for $F_2$ are useful for determining the output size of self-joins in databases, and for computing the surprise index
of a data sequence \cite{good1989c332}. Higher frequency moments are used to determine data skewness, which is important in parallel database applications \cite{dewitt1992practical}. 

We remark that for any fixed $p \neq 1$,\footnote{Note that there is a trivial $O(\log n)$-bit insertion-only $F_1$ estimation algorithm: keeping a counter for $\sum_t \Delta_t$.} including $p=0$, any deterministic insertion-only algorithm for $F_p$-estimation requires $\Omega(n)$ space \cite{AlonMS96,chakrabarti2016strong}. In contrast, we will show that randomized adversarially robust algorithms exist for all $p$, whose space complexity either matches or has a small multiplicative overhead over the best static randomized algorithms. 

We begin with several results on the problem of estimating distinct elements, or $F_0$ estimation. The first of them utilizes an optimized version of the sketch switching method to derive an upper bound. The result is an adversarially robust $F_0$ estimation algorithm whose complexity is only a $\Theta(\frac{1}{\eps} \log \eps^{-1} )$ factor larger than that of the optimal static (non-robust) algorithm \cite{blasiok2018optimal}.

\begin{theorem}[\small Robust Distinct Elements by Sketch Switch; see Theorem \ref{thm:distinct_elements_sketch_switching}]
	There is an algorithm which, when run on an adversarial insertion-only stream, with probability at least $1-\delta$ produces in every step $t \in [m]$ an estimate $R^t$ such that $R^t = (1 \pm \eps)\|f^{(t)}\|_0$ . The space used by the algorithm is 
	\[ O\left(\frac{\log(1/\eps)}{\eps} \left(\frac{\log \eps^{-1} + \log \delta^{-1} + \log \log n}{\eps^2} + \log n\right) \right).\]
\end{theorem}
The second result utilizes a different approach, by applying the computation paths method. The space complexity is slightly worse, which is a result of setting the failure probability $\delta < n^{-\frac{1}{\eps} \log n}$ for any given static algorithm. However, we introduce a new \textit{static} algorithm for $F_0$ estimation which has very small update-time dependency on $\delta$, and nearly optimal space complexity. As a result, by applying our computation paths method to this new static algorithm, we obtain an adversarially robust $F_0$ estimation algorithm with extremely fast update time (note that the update time of the above sketch switching algorithm would be $O(\eps^{-1} \log n)$ to obtain the same result, even for constant $\delta$).  

\begin{theorem}[\small Fast Robust Distinct Elements; see Theorem \ref{thm:fast}]
	There exists a streaming algorithm which, with probability $1-n^{-(C/\eps) \log n}$ for any constant $C \geq 1$, when run on an adversarially chosen insertion-only data stream, returns a $(1 \pm \eps)$ multiplicative estimate of the number of distinct elements in every step of the stream. The space required is $O(\frac{1}{\eps^3} \log^3  n )$, and the algorithm runs in  $O\left( \left(\log^2 \frac{\log n}{\eps} \right) \cdot \left(\log \log \frac{\log n}{\eps} \right)\right)$  worst case time per update.
\end{theorem}

The third result takes a different approach: it shows that under certain standard cryptographic assumptions, there exists an adversarially robust algorithm which asymptotically matches the space complexity of the best non-robust tracking algorithm for distinct elements. The cryptographic assumption is that an exponentially secure pseudorandom function exists (in practice one can take, for instance, AES as such a function).
While our other algorithms in this paper hold even against an adversary which is unbounded computationally, in this particular result we assume that the adversary runs in polynomial time. See Section \ref{sec:crypto} for more details. 

\begin{theorem}[\small Distinct Elements by Crypto Assumptions; see Theorem \ref{thm:distinct_elements_crypto}]\label{thm:cryptointro}
	In the random oracle model, there is an $F_0$-estimation (tracking) streaming algorithm in the adversarial setting, that for an approximation parameter $\varepsilon$ uses $O(\varepsilon^{-2}(\log 1/\varepsilon + \log \log n) + 	\log n)$ bits of memory, and succeeds with probability $3/4$.
	
	Moreover, given an exponentially secure pseudorandom function, and assuming the adversary has bounded running time of $n^c$, where $c$ is a constant, the random oracle can be replaced with a concrete function and the total memory is $O(\varepsilon^{-2}(\log 1/\varepsilon + \log \log n) + c\log n)$.
\end{theorem}
Here, the random oracle model means that the algorithm is given read access to an arbitrarily long string of random bits.

Our next set of results provides adversarially robust algorithms for $F_p$-estimation with $p > 0$. The following result concerns the case $0 < p \leq 2$. It was previously shown that for $p$ bounded away from one, $\Omega(n)$ space is required to deterministically estimate $\|f\|_p^p$, even in the insertion-only model \cite{AlonMS96,chakrabarti2016strong}. On the other hand, space-efficient non-robust randomized algorithms for $F_p$-estimation exist. 
We leverage these, along with an optimized version of the sketch switching technique to save a $\log n$ factor, and obtain the following. 

\begin{theorem}[\small Robust $F_p$-estimation for $0<p \leq 2$; see Theorem~\ref{thm:fpswitch}]\label{thm:introfp}
	Fix $0 < \eps,\delta \leq 1$ and $0 < p \leq 2$. There is a streaming algorithm in the insertion-only adversarial model which outputs in each step a value $R^t$ such that $R^t= (1 \pm \eps) \|f^{(t)}\|_p$ at every step $t \in [m]$, and succeeds with probability $1-\delta$. The algorithm uses $O( \eps^{-3} \log n \log \eps^{-1} (\log \eps^{-1} + 
	\log \delta^{-1} + \log \log n) )$ bits of space.
\end{theorem}
We remark that the space complexity of Theorem \ref{thm:introfp} is within a $\Theta(\eps^{-1} \log \eps^{-1})$ factor of the best known static (non-robust) algorithm \cite{blasiok2017continuous} . While for most values of $\delta$, the above theorem using sketch switching has better space complexity than the computation paths reduction, for the regime of very small failure probability $\delta$ it is actually preferable to use the latter, as we now state. 
\begin{theorem}[\small Robust $F_p$-estimation for small $\delta$; see Theorem \ref{thm:fpcompdelta}]
	Fix any $0 < \eps < 1$, $0 < p \leq 2$, and $\delta < n^{-C \frac{1}{\eps}\log n}$ for a sufficiently large constant $C>1$. There is a streaming algorithm for the insertion-only adversarial model which, with probability $1-\delta$,  successfully outputs in each step $t \in [m]$ a value $R^t$ such that $R^t= (1 \pm \eps) \|f^{(t)}\|_p$. The space used by the algorithm is $O\left( \frac{1}{\eps^2} \log n \log \delta^{-1} \right)$ bits.
\end{theorem}

In addition, we show that for turnstile streams with bounded $F_p$ \textit{flip number} (defined formally in Section \ref{sec:framework}), efficient adversarially robust algorithms exist. Roughly speaking, the $F_p$ flip number is the number of times that the $F_p$ moment changes by a factor of $(1+\eps)$. Our algorithms have extremely small failure probability of $\delta =n^{-\lambda}$, and have optimal space among turnstile algorithms with this value of $\delta$ \cite{jayram2013optimal}. 

\begin{theorem}[\small Robust $F_p$-Estimation in 
turnstile streams; See Theorem \ref{thm:fpturn}]
	Fix $0 < p \leq 2$ and let $\mathcal{S}_\lambda$ be the set of all turnstile streams with $F_p$ flip number at most $\lambda$. Then there is an adversarially robust streaming algorithm for the class $\mathcal{S}_\lambda$ of streams that, with probability $1-n^{-C\lambda}$ for any constant $C>0$, outputs in each step a value $R^t$ such that $R^t = (1 \pm \eps)\|f\|_p^p$. The space used by the algorithm is $O(\eps^{-2} \lambda \log^2 n )$.
\end{theorem}

The next result concerns $F_p$-estimation for $p > 2$. Here again, we provide an adversarially robust algorithm which is optimal up to a small multiplicative factor. This result applies the computation paths robustification method as a black box. Notably, a classic lower bound of \cite{bar2004information} shows that for $p > 2$, $\Omega(n^{1-2/p})$ space is required to estimate $\|f\|_p^p$ up to a constant factor (improved lower bounds have been provided since, e.g., \cite{li2013tight,ganguly2018high}).

\begin{theorem}[\small Robust $F_p$-estimation for $p > 2$; see Theorem \ref{thm:F_p_estim_largep}]
	Fix any $\eps>0$, and fix any $p > 2$. There is a streaming algorithm for the insertion-only adversarial model which, with probability $1-n^{-(c \log n) /\eps}$ for any constant $c>1$, successfully outputs, at each step $t \in [m]$, a value $R^t$ such that $R^t= (1 \pm \eps) \|f^{(t)}\|_p$. The space used by the algorithm is
	\[O\left(n^{1-2/p}  \left( \eps^{-3} \log^2 n + \eps^{-6/p} \left( \log^2 n \right)^{2/p} \log n \right) \right).\]
\end{theorem}

\paragraph{Attack on AMS.}
On the negative side, we demonstrate that the classic Alon-Matias-Szegedy sketch (AMS sketch) \cite{AlonMS96}, the first and perhaps most well-known $F_2$ estimation algorithm (which uses sub-polynomial space), is \textit{not} adversarially robust, even in the insertion-only setting.
Specifically, we demonstrate an adversary which, when run against the AMS sketch, fools the sketch into outputting a value which is not a $(1 \pm \eps)$ estimate of the $F_2$.  The non-robustness of standard static streaming algorithms, even under simple attacks, is a further motivation to design adversarially robust algorithms. 

In what follows, recall that the AMS sketch computes $S\cdot f$ throughout the stream, where $S \in \R^{t \times n}$ is a matrix of uniform $\{t^{-1/2},-t^{-1/2}\}$ random variables. The estimate of the $F_2$ is then the value $\|Sf\|_2^2$. 
\begin{theorem}[\small Attack on AMS sketch; see Theorem \ref{thm:AMS}]
	Let $S \in \R^{t \times n}$ be the AMS sketch with $1 \leq  t \leq n/c$ for some constant $c > 1$. There is an adversary which, with probability $99/100$, succeeds in forcing the estimate $\|Sf\|_2^2$ of the AMS sketch to not be a $(1 \pm 1/2)$ approximation of the true norm $\|f\|_2^2$. Moreover, the adversary needs to only make $O(t)$ stream updates before this occurs.
\end{theorem}

\paragraph{Heavy Hitters.}
We also show how our techniques can be used to solve the popular \textit{heavy-hitters} problem. Recall that the heavy-hitters problem tasks the streaming algorithm with returning a set $S$ containing all coordinates $i$ such that $|f_i| \geq \tau$, and containing no coordinates $j$ such that $|f_j| < \tau/2$, for some threshold $\tau$. Generally, the threshold $\tau$ is set to $\tau = \eps \|f\|_p$, which is known as the $L_p$ heavy hitters guarantee.

For $L_1$ heavy hitters in insertion-only streams, a deterministic $O(\frac{1}{\eps} \log n)$ space algorithm exists \cite{misra1982finding}. However, for $p>1$, specifically for the highly popular $p=2$, things become more complicated. Note that since we can have $\|f\|_2 \ll \|f\|_1$, the $L_2$ guarantee is substantially stronger.  For sketching-based turnstile algorithms, an $\Omega(n)$ lower bound for deterministic algorithms was previously known \cite{ganguly9}. Since $\|f\|_1 \leq \sqrt{n} \|f\|_2$, by setting $\eps = n^{-1/2}$, one can obtain a deterministic $O(\sqrt{n}\log n)$ space insertion-only $L_2$ heavy hitters algorithm. Recently, a lower bound of $\Omega(\sqrt{n})$ for deterministic insertion-only algorithms was given, demonstrating the near tightness of this result \cite{kpw20}. Thus, to develop a more efficient adversarially robust $L_2$ heavy hitters algorithm, we must employ randomness. 

Indeed, by utilizing our sketch switching techniques, we demonstrate an adversarially robust $L_2$ heavy hitters (tracking) algorithm which uses only an $O(\eps^{-1}\log \eps^{-1})$ factor more space than the best known static $L_2$ heavy hitters tracking algorithm \cite{braverman2017bptree}. Note that here the adversary sees the estimated set $S$ in every step.

\begin{theorem}[\small Robust $L_2$ heavy hitters: see Theorem  \ref{thm:HH}]
	Fix any $\eps > 0$. 
	There is a streaming algorithm in the adversarial insertion-only model which solves the $L_2$ heavy hitters problem in every step $t \in [m]$ with probability $1-n^{-C}$ (for any constant $C>1$). The algorithm uses $O(\frac{\log \eps^{-1}}{\eps^3}\log^2 n)$ bits of space. 
\end{theorem}

\paragraph{Entropy Estimation.}
Additionally, we demonstrate how our sketch switching techniques can be used to obtain robust algorithms for \textit{empirical Shannon Entropy estimation}. Here, the Shannon Entropy $H(f)$ of the stream is defined via $H(f) = -\sum_i \frac{|f_i|}{\|f\|_1} \log\left( \frac{|f_i|}{\|f\|_1}\right)$. Our results follow from an analysis of the exponential of $\alpha$-Renyi Entropy, which closely approximates the Shannon entropy, showing that the former cannot rapidly change too often within the stream. Our result is an  adversarially robust algorithm with space complexity only a small polylogarithmic factor larger than the best known static algorithms \cite{clifford2013simple,jayaram2019towards}.

\begin{theorem}[\small Robust Entropy Estimation; see Theorem \ref{thm:entropy}]\label{thm:entropyintro}
	There is an algorithm for $\eps$-additive approximation of the Shannon entropy in the insertion-only adversarial streaming model using $O(\frac{1}{\eps^4} \log^4 n (\log \log n + \log \eps^{-1}))$-bits of space in the random oracle model, and $O(\frac{1}{\eps^4} \log^6 n(\log \log n + \log \eps^{-1}))$-bits of space in the general insertion-only model. 
\end{theorem}

We remark that by making the same cryptographic assumption as in Theorem \ref{thm:cryptointro},  we can remove the random oracle assumption in \cite{jayaram2019towards} for correctness of the entropy algorithm in the static case. Then, by applying the same techniques which resulted in  Theorem \ref{thm:entropyintro}, we can obtain the same stated bound for entropy with a cryptographic assumption instead of a random oracle assumption. 

\paragraph{Bounded Deletion Streams.}
Lastly, we show that our techniques for $F_p$ moment estimation can be extended to data streams with a bounded number of deletions (negative updates). Specifically, we consider the bounded deletion model of \cite{jayaram2018data}. Formally, given some $\alpha \geq 1$, the model enforces the restriction that at all points $t \in [m]$ in the stream, we have $\|f^{(t)}\|_p^p \geq \frac{1}{\alpha} \|h^{(t)}\|_p^p$, where $h$ is the frequency vector of the stream with updates $u_i' = (a_i ,\Delta_i')$ where $\Delta_i' = |\Delta_i|$ (i.e., the absolute value stream). In other words,  the stream does not delete off an arbitrary amount of the $F_p$ weight that it adds over the course of the stream. 

We demonstrate that bounded deletion streams have the desirable property of having a small flip number, which, as noted earlier, is a measurement of how often the $F_p$ can change substantially (see Section \ref{sec:framework} for a formal definition). Using this property and our sketch switching technique, we obtain the following. 

\begin{theorem}[$F_p$-estimation for bounded deletion; see Theorem \ref{thm:bd}]
Fix $p \in [1,2]$, $\alpha \geq 1$, and any constant $C>1$. Then there is an adversarially robust $F_p$ estimation algorithm for $\alpha$-bounded deletion streams which, with probability $1-n^{-C}$, returns at each step $t \in [m]$ an estimate $R^t$ such that $R^t = (1 \pm \eps) \|f^{(t)}\|_p^p$. The space used by the algorithm is $O( \alpha \eps^{-(2+p)} \log^3 n)$.
\end{theorem}   

\subsection{Other Previous Work}
\label{subsec:related}
The need for studying adversarially robust streaming and sketching algorithms has been noted before in the literature. In particular, Gilbert et al.~\cite{gilbert2012recovering,gilbert2012reusable} motivated the adversarial model by giving applications and settings where it is impossible to assume that the queries made to a sketching algorithm are independent of the prior outputs of the algorithm, and the randomness used by the algorithm. One particularly important setting noted in \cite{gilbert2012reusable} is when the \textit{privacy} of the underlying data-set is a concern. 

In response to this, in \cite{HardtW13} the notion of adversarial robustness for \textit{linear} sketching algorithms was studied. Namely, it is shown how any function $g:\R^n \to \R$, defined by $g(x) = f(Ax)$ for some $A \in \R^{k \times n}$ and arbitrary $f: \R^k \to \R$ cannot approximate the $F_2$ moment $\|x\|_2^2$ of its input to an arbitrary polynomial factor in the presence of an adversary who is allowed to query $g(x_i)$ at polynomial many points (unless $k$ is large). Since one can insert and delete off each $x_i$ in a turnstile stream, this demonstrates a strong lower bound for adversarially robust turnstile linear sketching algorithms, at least when the stream updates are allowed to be real numbers.  Moreover, under certain conditions it has been demonstrated that all turnstile algorithms can be transformed into linear sketches \cite{li2014turnstile,ai2016new,KallaugherPrice20}. We point out, however, that this equivalence holds only for classes of \textit{static} streams, and therefore does not immediately have any consequence for adversarial streams. The work of \cite{HardtW13} also points out a connection to differential privacy. 

We remark that other work has observed the danger inherent in allowing adversarial queries to a randomized sketch with only a static guarantee, see Ahn et al.~\cite{ahn2012analyzing,ahn2012graph}. However, the motivation of these works is slightly different, and their setting not fully adversarial. Mironov et al.~\cite{mironov2011sketching} considered adversarial robustness of sketching in a distributed, \textit{multi-player} model, which is incomparable to the centralized streaming problem considered in this work. Finally, Goldwasser et al.~\cite{goldwasser2019pseudo} asked if there are randomized streaming algorithms whose output is independent of its randomness, making such algorithms natural candidates for adversarial robustness; unfortunately a number of their results are negative, while their
upper bounds do not apply to the problems studied here.  

\subsection{Subsequent Work and Open Questions}
\label{subsec:future}
Based on this paper, multiple very recent follow-up works have improved upon the space efficiency of our robustification techniques for different settings. Hassidim et al.~\cite{HassidimKMMS20} used techniques from differential privacy to obtain a generic robustification framework in the same mold as ours, where the dependency on the flip number is the improved $\sqrt{\lambda}$ as opposed to linear in $\lambda$ -- the exact bound includes other $\poly((\log n) / \epsilon)$ factors. Similar to our construction, they run multiple independent copies of the static algorithm $A$ with independent randomness and feed the input stream to all of the copies. Unlike our construction, when a query comes, they aggregate the responses from the copies in a way that protects the internal randomness of each of the copies using differential privacy. Using their framework, one may construct an adversarially robust algorithm for $F_p$-moment estimation that uses $\widetilde{O}(\frac{\log^4 n}{\epsilon^{2.5}})$ bits of memory for any $p \in [0,2]$. This improves over our $\widetilde{O}(\frac{\log n}{\epsilon^{3}})$ bound for interesting parameter regimes.

Woodruff and Zhou \cite{WoodruffZ20} obtained further improvements for a number of streaming problems (such as $F_p$-estimation, entropy, heavy hitters) which in some cases are nearly optimal even for the static case. For example, they give an adversarially robust algorithm for $F_p$-moment estimation that uses $\widetilde{O}(\frac{\log n}{\epsilon^2})$ bits of memory for any $p \in [0,2]$. This improves upon both our work and \cite{HassidimKMMS20}. Interestingly, the way they achieve this leads them to a new class of (classical) streaming algorithms they call difference estimators, which turn out to be useful also in the sliding window (classical) model. Subsequently, Attias and el.~\cite{ACSS21} 
combined the differential privacy based techniques of \cite{HassidimKMMS20} with the difference estimators of \cite{WoodruffZ20} to obtain a ``best of both worlds'' result with improved bounds for turnstile streams. 

It was shown by Kaplan et al.~\cite{KMNS21} that the $\sqrt{\lambda}$-type space overhead is tight for some streaming problems; they proved this for a streaming variant of the Adaptive Data Analysis problem, showing also that its space complexity is polylogarithmic in the static setting and polynomial in the adversarially robust setting. This is the first known example of such a large separation between static and adversarially robust streaming. Another interesting work by Menuhin and Naor \cite{MN21} shows that card guessing performance with memory constraints may be exponentially worse  against an adaptive adversarial dealer versus a static one.

For core problems in the streaming literature like $F_p$-estimation in the turnstile model (allowing insertions and deletions), it is not known whether such a separation exists.
However, there is a substantial gap between the space complexity of the static case and the best known algorithms for the adversarially robust case. For static turnstile $F_p$-estimation, the space complexity is polylogarithmic in $n$ when $p \leq 2$. 
In the adversarially robust setting, the best known results are much weaker, and involve polynomial dependence in the stream length $m$. As the above robustification techniques induce a $\sqrt{\lambda}$ overhead in the space complexity, and $\lambda=m$ for turnstile $F_p$-estimation, these techniques cannot obtain space bounds better than some $O(\sqrt{m})$ in general. Recently, Ben-Eliezer et al.~\cite{BEO21} used a hybrid approach combining the differential privacy based framework of \cite{HassidimKMMS20} with classical results in sparse recovery to obtain improved space bounds for this problem; the dependence in $m$ is for example $\tilde{O}(m^{1/3})$ when $p=0$ and $\tilde{O}(m^{2/5})$ when $p=2$.
This large gap in the best known space requirements, despite the fact that no space complexity separations between static and robust algorithms are known, leads to the following natural question (see \cite{jayaram2021thesis}, \cite{BEO21}):
\begin{quotation}
\begin{center}
\emph{What is the space complexity of adversarially robust\ \ \ \ \\\ \ \ \  $F_p$-estimation under the turnstile streaming model?}
\end{center}
\end{quotation}

Many problems remain open, mainly for achieving optimal bounds for all known streaming problems in the adversarial setting. It is also interesting to determine which types of existing algorithms are inherently adversarially robust. Remarkably, Braverman et al.~\cite{braverman2021adversarial} showed that popular techniques such as merge and reduce and row sampling can be robust ``for free'', implying robustness guarantees for many types of existing algorithms for streaming, regression, low rank approximation, and various other problems. Unlike our setting, which considers algorithms with a scalar output (i.e., an output which is typically a single real number), many of these problems produce a higher-dimensional vector output. It will be interesting to investigate what sorts of extensions of our flip number definition may be relevant in high dimensions, and to find suitable applications for such a generalized flip number notion.

A first result in this flavor has very recently been established by Chakrabarti et al.~\cite{RobustColoringCGS2021}, who considered the problem of coloring a graph in the semi-streaming model. They proved that coloring with few colors requires substantially more space in the adversarial model compared to the static one; for example, $O(\Delta)$ colors require $\Omega(n \Delta)$ space in the robust setting but only $O(n)$ space in the static setting (see \cite{ACKcoloring19}). They then provided adversarially robust algorithms for this problem, including one algorithm based on our main technique, sketch switching.

\section{Preliminaries}
For $p >0$, the $L_p$ norm\footnote{Note that this is only truly a norm for $p\geq 1$.} of a vector $f \in \R^n$ is given by $\|f\|_p = \left( \sum_{i=1}^n |f_i|^p \right)^{1/p}$. The $p$-th moment, denoted by $F_p$, is given by $F_p = L_p^{p}$, or $F_p = \sum_i |f_i|^p$. For $p=0$, we define $F_0$ to be the number of non-zero coordinates in $f$, namely $F_0 = \|f\|_0 = |\{i \; : \; f_i \neq 0 \}|$. Notice that this coincides with defining $0^0 = 0$ in the prior definition of $F_p$. The $F_0$ moment is also known as the number of distinct elements. 
For reals $a,b \in \R$ and $\eps > 0$, we write $a = (1 \pm \eps) b$ or $a \in (1 \pm \eps) b$ to denote the containment $a \in [(1-\eps) b, (1+\eps)b]$. Throughout, we will often assume that our error parameter $\eps>0$ is smaller than some absolute constant $\eps_0$ which does not depend on any of the other parameters of the problem. 

A  stream of length $m$ over a domain $[n]$ is a sequence of updates $(a_1,\Delta_1),(a_2,\Delta_2)\dots,(a_m,\Delta_m)$ where $a_t \in [n]$ and $\Delta_t \in \Z$. 
The \textit{frequency vector} $f \in \R^n$ of the stream is the vector with $i^{\text{th}}$ coordinate $f_i = \sum_{t: a_t = i} \Delta_t$. Let $f^{(j)}$ be the frequency vector restricted to the first $j$ updates, namely $f_i^{(j)} = \sum_{t \leq j : a_t = i} \Delta_t$. It is assumed at all intermediate points $t \in [m]$ in the stream that $\|f^{(t)}\|_\infty \leq M$, and $\log(mM) = \Theta(\log n )$. Notice in particular that this bounds $|\Delta_t| \leq 2M$ for each $t$. 

The general model as defined above is known as the \textit{turnstile} model of streaming. Another commonly studied model of streaming is the \textit{insertion-only} model, where it is assumed that $\Delta_t > 0$ for each $t =1,\dots,m$. The insertion-only model is often presented with the following equivalent and simplified definition: an insertion-only stream is given by a sequence $a_1,a_2,\dots,a_m \in [n]$, and the frequency vector $f \in \R^n$ is given by $f_i = |\{ j \in [m] : a_j = i\}|$. Since we will sometimes consider data streams with deletions (negative updates), in this work, we will use the former definition, where updates are pairs $(a_t,\Delta_t) \in [n]\times \Z$. In this paper, the space of a streaming algorithm is measured in bits, and the update time of a streaming algorithm is measured in the RAM model, where arithmetic operations on $O(\log n)$-bit integers can be done in $O(1)$ time. Throughout the paper we will almost always assume that the output (at any time) of the algorithms we discuss is represented by $O(\log n)$ bits; since we are generally interested in $(1+\eps)$-approximation where $\log(1/\eps) = O(\log n)$, any algorithm with higher bit precision can be replaced by one that only outputs the most significant $O(\log n)$ bits at any step, without majorly affecting any of the results. The only exception where the output requires more than $O(\log n)$ bits is for $F_2$-heavy hitters; here, a total of $O(\log n/\eps^2)$ bits are generally required to store all heavy hitters. 

The \textit{random-oracle} model of streaming is the model where the streaming algorithm is allowed random (read-only) access to an arbitrarily long string of random bits. In other words, the space complexity of the algorithm is not charged for storing random bits. We remark that while nearly all lower bounds for streaming algorithms hold even in the random oracle model, most of our results (except for one of our results for entropy estimation and part of our cryptographic results) do not require a random oracle.


Finally, given a vector $x \in \R^n$, the \textit{empirical Shannon Entropy} $H(x)$ is defined via $H(x) = -\sum_i |x_i|/\|x\|_1 \log\left( |x_i|/\|x\|_1\right)$. For $\alpha > 0$, the $\alpha$-Renyi Entropy $H_\alpha(x)$ of $x$ is given by the value $H_\alpha(x) = \log( \|x\|_\alpha^\alpha/\|x\|_1^\alpha)/(1-\alpha)$.

\subsection{Tracking Algorithms}\label{sec:tracking}
The robust streaming algorithms we design in this paper satisfy the \textit{tracking} guarantee. Namely, they must output a response to a query at every step in time $t \in [m]$. For the case of estimation queries, this tracking guarantee is known as strong tracking.

\begin{definition}[Strong tracking]
Let $f^{(1)},f^{(2)},\dots,f^{(m)}$ be the  frequency vectors of a stream $(a_1, \Delta_1), \ldots, (a_m, \Delta_m)$, and let $g:\R^n \to \R$ be a function on frequency vectors. A randomized algorithm $\mathcal{A}$ is said to provide $(\eps,\delta)$-\textit{strong $g$-tracking} if at each step $t \in [m]$ it outputs an estimate $R_t$ such that 
\[  |R_t - g(f^{(t)}) | \leq \eps |g(f^{(t)})| \]    for all $t \in [m]$ with probability at least $1-\delta$.
\end{definition}
In contrast, \textit{weak tracking} replaces the error term $\eps |g(f^{(t)})|$ by $\max_{t' \in [m]} \eps \cdot |g(f^{(t')})|$. However, for the purposes of this paper, we will not need to consider weak tracking. We now state two results for strong tracking of $F_p$ moments for $p \in [0,2]$. Both results are for the static setting, i.e., for a stream fixed in advance (and not for the adaptive adversarial setting that we consider).
\begin{lemma}[\cite{blasiok2017continuous}]\label{lem:fpstrong}
For $0<p\leq 2$, there is an insertion-only streaming algorithm which provides $(\eps,\delta)$-strong $F_p$-tracking using $O(\frac{\log n}{\eps^2} (\log \eps^{-1} + \log \delta^{-1} + \log \log n))$ bits of space. 
\end{lemma}

\begin{lemma}[\cite{blasiok2018optimal}]\label{lem:f0strong}
There is an insertion-only streaming algorithm which provides $(\eps,\delta)$-strong $F_0$-tracking using $O(\frac{\log \log n + \log \delta^{-1}}{\eps^2}  + \log n)$ bits of space. 
\end{lemma}

\subsection{Roadmap}
In Section \ref{sec:framework}, we introduce our two general techniques for transforming static streaming algorithms into adversarially robust algorithms. In Section \ref{sec:Fp}, we give our results on estimation of $F_p$ moments, and in Section \ref{sec:F0} we give our algorithms for adversarially robust distinct elements estimation. Next, in Section \ref{sec:HH}, we introduce our robust $L_2$ heavy hitters algorithm, and in Section \ref{sec:Entropy} we give our entropy estimation algorithm. In Section \ref{sec:BoundedDel}, we provide our algorithms for $F_p$ moment estimation in the bounded deletion model. In Section \ref{sec:AMS}, we give our adversarial attack on the AMS sketch. Finally, in Section \ref{sec:crypto}, we give our algorithm for optimal space distinct elements estimation under cryptographic assumptions.

\section{Tools for Robustness}
\label{sec:framework}
In this section, we establish two methods, \emph{sketch switching} and \emph{computation paths}, allowing one to convert an approximation algorithm for any sufficiently well-behaved streaming problem to an adversarially robust one for the same problem. 
The central definition of a \emph{flip number}, bounds the number of major (multiplicative) changes in the algorithm's output along the stream.
As we shall see, a small flip number allows for efficient transformation of non-robust algorithms into robust ones.\footnote{The notion of flip number we define here also plays a central role in subsequent works (\cite{HassidimKMMS20}, \cite{WoodruffZ20}); for example, the main contribution of the former is a generic robustification technique with an improved (square root type instead of linear) dependence in the flip number. The latter improves the $\poly(1/\epsilon)$ dependence on the flip number.}

\subsection{Flip Number}
\begin{definition}[flip number]
\label{def:flip}
Let $\eps \geq 0$ and $m \in \N$, and let $\bar y = (y_0, y_1, \ldots, y_m)$ be any sequence of real numbers. The \emph{$\eps$-flip number} $\lambda_\eps(\bar y)$ of $\bar y$ is the maximum $k \in \N$ for which there exist $0 \leq i_1 < \ldots < i_k \leq m$ so that $y_{i_{j-1}} \notin (1 \pm \eps) y_{i_{j}}$ for every $j=2,3,\ldots,k$. 

Fix a function $g \colon \R^n \to \R$ and a class $\mathcal{C} \subseteq ([n] \times \Z)^m$ of stream updates. The $(\eps, m)$-flip number $\lambda_{\eps, m}(g)$ of $g$ over $\mathcal{C}$ is the maximum, over all sequences $((a_1, \Delta_1), \ldots, (a_m, \Delta_m)) \in \mathcal{C}$, of the $\eps$-flip number of the sequence $\bar{y} = (y_0,y_1, \ldots, y_m)$ defined by $y_i = g(f^{(i)})$ for any $0 \leq i \leq m$, where as usual $f^{(i)}$ is the frequency vector after stream updates $(a_1, \Delta_1), \ldots, (a_i, \Delta_i)$ (and $f^{(0)}$ is the $n$-dimensional zeros vector).
\end{definition}
The class $\mathcal{C}$ may represent, for instance, the subset of all insertion-only streams, or bounded-deletion streams. For the rest of this section, we shall assume $\cal C$ to be fixed, and consider the flip number of $g$ with respect to this choice of $\cal C$.\footnote{A somewhat reminiscent definition, of an \emph{unvarying algorithm}, was studied by \cite{DworkNPR10} (see Definition 5.2 there) in the context of differential privacy. While their definition also refers to a situation where the output undergoes major changes only a few times, both the motivation and the precise technical details of their definition are different from ours.}

Note that the flip number is clearly monotone in $\eps$: namely $\lambda_{\eps',m}(g) \geq \lambda_{\eps,m}(g)$ if $\eps' < \eps$. 
One useful property of the flip number is that it is nicely preserved under approximations. As we show, this can be used to effectively construct approximating sequences whose \emph{$0$-flip number} is bounded as a function of the $\eps$-flip number of the original sequence. This is summarized in the following lemma.
\begin{lemma}
\label{lemma:flip_number_for_sketch_switching}
Fix $0 < \eps < 1$. Suppose that $\bar u = (u_0, \ldots, u_m)$, $\bar v = (v_0, \ldots, v_m)$, $\bar w = (w_0, \ldots, w_m)$ are three sequences of real numbers, satisfying the following: 
\begin{itemize}
    \item For any $0 \leq i \leq m$, $v_i = (1 \pm \eps/8) u_i$.
    \item $w_0 = v_0$, and for any $i > 0$, if $w_{i-1} = (1 \pm \eps/2) v_i$ then $w_i = w_{i-1}$, and otherwise $w_i = v_i$.
\end{itemize}
Then $w_i = (1 \pm \eps) u_i$ for any $0 \leq i \leq m$, and moreover, $\lambda_0(\bar w) \leq \lambda_{\eps/8} (\bar u)$.

In particular, if (in the language of Definition \ref{def:flip}) $u_0 = g(f^{(0)}), u_1 = g(f^{(1)}),\ldots, u_m=g(f^{(m)})$ for a sequence of updates $((a_1, \Delta_1), \ldots, (a_m, \Delta_m)) \in \mathcal{C}$, then $\lambda_0(\bar w) \leq \lambda_{\eps/8, m}(g)$.
\end{lemma}
\begin{proof}
The first statement, that $w_i = (1 \pm \eps) u_i$ for any $i$, follows immediately since $v_i = (1 \pm \eps/8)u_i$ and $w_i = (1 \pm \eps/2)v_i$ and since $\eps < 1$. The third statement follows by definition from the second one. It thus remains to prove that $\lambda_0(\bar w) \leq \lambda_{\eps/8} (\bar u)$.

Let $i_1 = 0$ and let $i_2, i_3, \ldots, i_k$ be the collection of all values $i \in [m]$ for which $w_{i-1} \neq w_{i}$. Note that $k = \lambda_0(\bar w)$ and that $v_{i_{j-1}} = w_{i_{j-1}} = w_{i_{j-1} + 1} = \cdots = w_{i_{j} - 1} \neq v_{i_j}$ for any $j=2,\ldots,k$. 
We now claim that for every $j$ in this range, $u_{i_{j-1}} \notin (1 \pm \eps/8) u_{i_j}$. 
This would show that $k \leq \lambda_{\eps/8}(\bar u)$ and conclude the proof.

Indeed, fixing any such $j$, we either have $v_{i_{j-1}} > (1+\eps/2)v_{i_j}$, or $v_{i_{j-1}} < (1-\eps/2)v_{i_j}$. In the first case (assuming $u_{i_j} \neq 0$, as the case $u_{i_{j}} = 0$ is trivial),
$$
\frac{u_{i_{j-1}}}{u_{i_j}} \geq \frac{v_{i_{j-1}}/(1+\frac{\eps}{8})}{v_{i_{j}}/(1-\frac{\eps}{8})} \geq \left(1+\frac{\eps}{2}\right) \cdot \frac{1-\frac{\eps}{8}}{1+\frac{\eps}{8}} > 1+\frac{\eps}{8}\,.
$$
In the second case, an analogous computation gives $u_{i_{j-1}}/u_{i_j} < 1-\eps/8$.
\end{proof}

Note that the flip number of a function $g$ critically depends on the model in which we work, as the maximum is taken over all sequences of \emph{possible stream updates}; for insertion-only streams, the set of all such sequences is more limited than in the general turnstile model, and correspondingly many streaming problems have much smaller flip number when restricted to the insertion-only model.
We now give an example of a class of functions with bounded flip number. 

\begin{proposition}\label{prop:monfpflip}
Let $g: \R^n \to \R$ be any monotone function, meaning that $g(x) \geq g(y)$ if $x_i \geq y_i$ for each $i \in [n]$. Assume further that 
$g(x) \geq T^{-1}$ for all $x > 0$, and $g(M \cdot \vec{1}) \leq T$, where $M$ is a bound on the entries of the frequency vector and $\vec{1}$ is the all $1$'s vector. Then the flip number of $g$ 
in the insertion-only streaming model is $\lambda_{\eps,m}(g) = O(\frac{1}{\eps}\log T)$.
\end{proposition}
\begin{proof}
To see this, note that $g(f^{(1)}) \geq T^{-1}$, and $g(f^{(m)}) \leq g(\vec{1} \cdot M) \leq T$. Since the stream has only positive updates, $g(f^{(0)}) \leq g(f^{(1)}) \leq \dots \leq g(f^{(m)})$. 
Let $1 \leq y_1< y_2 < \dots< y_k \in [m]$ be any maximal increasing sequence of time steps such that $g(f^{(y_i)}) < (1 - \eps)g(f^{(y_{i+1})})$ for each $i \in [k-1]$. Note that restricting to $y_1 \geq 1$ only excludes the $0$-th step, so the flip number is at most $k+1$.  Then the value of $g$ increases by a $\frac{1}{1- \eps}$ factor after each step $y_i$. Since there are at most $O(\frac{1}{\eps}\log T)$ powers of $\frac{1}{1- \eps}$ between $T^{-1}$ and $T$,  by the pigeonhole principle if $k > \frac{C}{\eps}\log(T)$ for a sufficiently large constant $C$, then at least two values must satisfy $(\frac{1}{1- \eps})^{j} \leq g(f^{(y_i)}) \leq g(f^{(y_{i+1})}) \leq (\frac{1}{1- \eps})^{j+1}$ for some $j$, which is a contradiction. 
\end{proof}
Note that a special case of the above are the $F_p$ moments of a data stream. Recall here $\|x\|_0 = |\{ i : x_i \neq 0 \} |$ is the number of non-zero elements in a vector $x$. For what follows, recall that the stream length is $m=O(\poly(n))$.

\begin{corollary}\label{cor:fpflip}
Let $p \geq 0$. The $(\eps,m)$-flip number of $\|x\|_p^p$ in the insertion-only streaming model is $\lambda_{\eps,m}(\|\cdot \|_p^p) = O(\frac{1}{\eps}\log n)$ for $p \leq 2$, and $\lambda_{\eps,m}(\|\cdot \|_p^p) = O(\frac{p}{\eps}\log n)$ for $p>2$. For $p=0$, we also have $\lambda_{\eps,m}(\|\cdot \|_0) = O(\frac{1}{\eps}\log m)$.
\end{corollary}
\begin{proof}
We have $\|\vec{0}\|_p^p = 0$, $\|z\|_p^p \geq 1$ for any non-zero $z \in \Z$, and $\|f^{(m)}\|_p^p \leq M^pn \leq n^{1+cp}$ for some constant $c$, where the second to last inequality holds because $\|f\|_\infty \leq M$ for some $M = \poly(n)$ is assumed at all points in the streaming model. The result then follows from applying Proposition \ref{prop:monfpflip} with $T = n^{c \cdot \max\{p,1\}}$. The last statement for $p=0$ follows since $\|f^{(m)}\|_0$ either remains unchanged or increases by one after any single insertion.
\end{proof}
Another special case of Proposition \ref{prop:monfpflip} concerns the \textit{cascaded norms} of insertion-only data streams \cite{jayram2009data}. Here, the frequency vector $f$ is replaced with a matrix $A \in \Z^{n \times d}$, which receives coordinate-wise updates in the same fashion, and the $(p,k)$ cascaded norm of $A$ is given by $\|A\|_{(p,k)} =(\sum_{i}(\sum_j |A_{i,j}|^{k} )^{p/k} )^{1/p} $. In other words, $\|A\|_{(p,k)}$ is the result of first taking the $L_k$ norm of the rows of $A$, and then taking the $L_p$ norm of the result. Proposition \ref{prop:monfpflip} similarly holds with $T = \poly(n)$ in the insertion-only model, and therefore the black-box reduction techniques introduced in the following sections are also applicable to these norms (using e.g., the cascaded algorithms of \cite{jayram2009data}). 

Having a small flip number is very useful for robustness, as our next two robustification techniques demonstrate.

\subsection{The Sketch Switching Technique}

Our first technique is called \emph{sketch switching}, and is described in Algorithm \ref{alg:sample}. The technique maintains multiple instances of a static strong tracking algorithm, where at any given time only one of the instances is ``active''. The idea is to change the current output of the algorithm very rarely. Specifically, as long as the current output is a good enough multiplicative approximation of the estimate of the active instance, the estimate we give to the adversary does not change, and the current instance remains active. As soon as this approximation guarantee is not satisfied, we update the output given to the adversary, deactivate our current instance, and activate the next one in line. By carefully exposing the randomness of our multiple instances, we show that the strong tracking guarantee (which a priori holds only in the static setting) can be carried into the robust setting. 
By Lemma \ref{lemma:flip_number_for_sketch_switching}, the required number of instances, which corresponds to the $0$-flip number of the outputs provided to the adversary, is controlled by the $(\Theta(\eps), m)$-flip number of the problem.

\begin{algorithm}[!ht]
	$\lambda \leftarrow \lambda_{\eps/8,m}(g)$ \\
	Initialize independent instances $A_1, \ldots, A_\lambda$ of $(\frac{\eps}{8}, \frac{\delta}{\lambda})$-strong $g$-tracking algorithm \\
	$\rho \leftarrow 1$\\
	$\tilde{g} \leftarrow g(\vec{0})$\\ 
	\While{new stream update $(a_k,\Delta_k)$}{
	Insert update $(a_k,\Delta_k)$ into each algorithm $A_1,\dots,A_\lambda$ \\
	$y \leftarrow$ current output of $A_\rho$ \\
	\uIf{$\tilde{g} \notin (1 \pm \eps/2)y$}{
	$\tilde{g} \leftarrow y$ \\ 
	$\rho \leftarrow \rho+1$
	}
Output estimate $\tilde{g}$
	}
\caption{Adversarially Robust $g$-estimation by Sketch Switching}
\label{alg:sample}
\end{algorithm}

\begin{lemma}[Sketch Switching]\label{lem:sketchswitch}
Fix any function $g: \R^n \to \R$ and let $A$ be a streaming algorithm that for any $0 < \eps < 1$ and $\delta > 0$ uses space $L(\eps,\delta)$, and satisfies the $(\eps,\delta)$-strong $g$-tracking property on the frequency vectors $f^{(1)}, \ldots, f^{(m)}$ of any particular fixed stream. Then Algorithm \ref{alg:sample} is an adversarially robust algorithm for $(1+\eps)$-approximating $g(f^{(t)})$ at every step $t \in [m]$ with success probability $1-\delta$, whose space is $O\left(L(\eps/8,\delta/\lambda) \cdot \lambda\right)$, where $\lambda = \lambda_{\eps/8, m}(g)$.
\end{lemma}

The proof is by induction and we start by giving its main intuition. By Yao's minimax principle, one may assume the adversary is deterministic (but adaptive). 
Consider the point in time $t_\rho$ where the output $y_\rho$ of the $\rho$-th instance, $A_\rho$, is first sent to the adversary. From this point on, the output displayed to the adversary is $y_\rho$, whereas the next instance $A_{\rho + 1}$ continues to run and update its output internally (without displaying it to the adversary). 
Let $t_{\rho+1}$ be the first point in time where the (internal) output of $A_{\rho+1}$ substantially differs from $y_\rho$; denote this output by $y_{\rho+1}$, and set the value displayed to the adversary to $y_{\rho+1}$. 
The crucial observation is that we only need to apply the static tracking guarantee for a single specific input sequence in order to ensure that $y_{\rho}$ is a good approximation of our function $f$ at any time between $t_\rho$ and $t_{\rho+1}-1$.
The said input sequence consists of all inputs provided by the adversary until time $t_\rho$, concatenated with the sequence of inputs that the adversary would send if it were to see the fixed output $y_\rho$ for $m-t_\rho$ times afterward.

Now, how many times will the active instance change during this process? Our choice of parameters in the algorithm ensures that each such change can happen only if the value of the function $f$ itself has changed by some $1 \pm \eps/8$. Thus, the number of instances required is bounded by $\lambda_{\eps/8,m}(f)$.

\begin{proof}
Note that for a fixed randomized algorithm $\mathcal{A}$ we can assume the adversary against $\mathcal{A}$ is deterministic without loss of generality (in our case, $\cal A$ refers to Algorithm \ref{alg:sample}). This is because given a randomized adversary and algorithm, if the adversary succeeds with probability greater than $\delta$ in fooling the algorithm, then by a simple averaging argument, there must exist a fixing of the random bits of the adversary which fools $\mathcal{A}$ with probability greater than $\delta$ over the coin flips of $\mathcal{A}$. Note also here that conditioned on a fixing of the randomness for both the algorithm and adversary, the entire stream and behavior of both parties is fixed.

We thus start by fixing such a string of randomness for the adversary, which makes it deterministic. As a result, suppose that $y_i$ is the output of the streaming algorithm
in step $i$. Then given $y_1,y_2,\dots,y_k$ and the stream updates $(a_1,\Delta_1),\dots,(a_k,\Delta_k)$ so far, the next stream update $(a_{k+1},\Delta_{k+1})$ is deterministically fixed. We stress that the randomness of the algorithm is not fixed at this point; we will gradually reveal it along the proof. 

Let $\lambda = \lambda_{\eps/8, m}(g)$ and let $A_1, \ldots, A_{\lambda}$ be the $\lambda$ independent instances of an $(\eps/8, \delta/\lambda)$-strong tracking algorithm for $g$. Since $\delta_0 = \delta/\lambda$, later on we will be able to union bound over the assumption that for all $\rho \in [\lambda]$, $A_i$ satisfies strong tracking on some fixed stream (to be revealed along the proof); the stream corresponding to $A_\rho$ will generally be different than that corresponding to $\rho'$ for $\rho \neq \rho'$. 

First, let us fix the randomness of the first instance, $A_1$. Let $u_1^1,u_2^1,\dots,u_m^1$ be the updates $u^1_j = (a_j, \Delta_j)$ that the adversary would make if $\cal A$ were to output $y_0 = g(\vec{0})$ at every time step, and let $f^{(t),1}$ be the stream vector after updates $u_1^1,\dots,u_t^1$. Let $A_1(t)$ be the output of algorithm $A_1$ at time $t$ of the stream $u^1_1,u^1_2,\dots,u_t^1$. Let $t_1 \in [m]$ be the first time step such that $y_0 \notin (1 \pm \eps/2) A_1(t_1)$, if exists (if not we can set, say, $t_1 = m+1$).  At time $t = t_1$, we change our output to $y_1 = A_1(t_1)$. Assuming that $A_{1}$ satisfies strong tracking for $g$ with approximation parameter $\eps/8$ with respect to the fixed stream of updates $u_1^1, \ldots, u_m^1$ (which holds with probability at least $1-\delta/\lambda$), we know that $A_1(t) = (1 \pm \eps/8) g(f^{(t)})$ for each $t < t_1$ and that $y_0 = (1 \pm \eps/2) A_1(t)$. Thus, by the first part of Lemma \ref{lemma:flip_number_for_sketch_switching}, $y_0 = (1 \pm \eps) g(f^{(t)})$ for any $0 \leq t < t_1$.
Furthermore, by the strong tracking, at time $t=t_1$ the output we provide $y_1 = A_1(t_1)$ is a $(1\pm\eps/8)$-approximation of the desired value $g(f^{(t_1)})$.

At this point, $\cal A$ ``switches'' to the instance $A_2$, and presents $y_1$ as its output as long as $y_1 = (1 \pm \eps/2) A_2(t)$. Recall that randomness of the adversary is already fixed, and consider the sequence of updates obtained by concatenating $u_1^1, \ldots, u_{t_1}^{1}$ as defined above (these are the updates already sent by the adversary) with the sequence $u^2_{t_1+1}, \ldots, u^2_{m}$ to be sent by the adversary if the output from time $t=t_1$ onwards would always be $y_1$. We condition on the $\eps/8$-strong $g$-tracking guarantee on $A_2$ holding for this fixed sequence of updates, noting that this is the point where the randomness of $A_2$ is revealed. Set $t = t_2$ as the first value of $t$ (if exists) for which $A_2(t) = (1 \pm \eps/2) y_1$ does not hold. We now have, similarly to above, $y_1 = (1 \pm \eps) g(f^{(t)})$ for any $t_1 \leq t < t_2$, and $y_2 = (1 \pm \eps/8) g(f^{(t_2)})$.

The same reasoning can be applied inductively for $A_\rho$, for any $\rho \in [\lambda]$, to get that (provided $\eps/8$-strong $g$-tracking holds for $A_{\rho}$) at any given time, the current output we provide to the adversary $y_\rho$ is within a $(1 \pm \eps)$-multiplicative factor of the correct output for any of the time steps $t = t_{\rho}, t_{\rho}+1, \ldots, \min\{t_{\rho+1}-1,m\}$. 
Taking a union bound, we get that with probability at least $1-\delta$, all instances provide $\eps/8$-tracking (each for its respective fixed sequence), yielding the desired $(1\pm \eps)$-approximation of our algorithm.

It remains to verify that this strategy succeeds in handling all $m$ elements of the stream (and does not exhaust its pool of algorithm instances before then). Indeed, this follows immediately from Lemma \ref{lemma:flip_number_for_sketch_switching} applied with $\bar u = ((g(f^{(0)}), \ldots, g(f^{(m)}))$, $\bar v = (g(f^{(0)}), A_1(1), \ldots, A_1(t_1), A_2(t_1+1), \ldots, A_2(t_2), \ldots)$, and $\bar w$ being the output that our algorithm $\cal A$ provides ($y_0 = g(f^{(0)})$ until time $t_1-1$, then $y_1$ until time $t_2-1$, and so on). Observe that indeed $\bar w$ was generated from $v$ exactly as described in the statement of Lemma \ref{lemma:flip_number_for_sketch_switching}.
\end{proof}

\subsection{The Bounded Computation Paths Technique}
With our sketch switching technique, we showed that maintaining multiple instances of a non-robust algorithm to estimate a function $g$, and switching between them when the rounded output changes, is a recipe for a robust algorithm to estimate $g$. We next provide another recipe, which keeps only one instance, whose success probability for any fixed stream is very high; it relies on the fact that if the flip number is small, then the total number of fixed streams that we should need to handle is also relatively small, and we will be able to union bound over all of them. Specifically, we show that any non-robust algorithm for a function with bounded flip number can be modified into an adversarially robust one by setting the failure probability $\delta$ small enough. 

\begin{lemma}[Computation Paths]\label{lem:computationpaths}
Fix $g \colon \R^n \to \R$ and suppose that the output of $g$ uses $\log T$ bits of precision (see Remark \ref{remark:bit_prec}).
Let $A$ be a streaming algorithm that for any $\eps, \delta > 0$ satisfies the $(\eps, \delta)$-strong $g$-tracking property on the frequency vectors $f^{(1)}, \ldots, f^{(m)}$ of any particular fixed stream. 
Then there is a streaming algorithm $A'$ satisfying the following.
\begin{enumerate}
    \item $A'$ is an adversarially robust algorithm for $(1+\eps)$-approximating $g(f^{(t)})$ in all steps $t \in [m]$, with success probability $1-\delta$.
    \item The space complexity and running time of $A'$ as above (with parameters $\eps$ and $\delta$) are of the same order as the space and time of running $A$ in the static setting with parameters $\eps/8$ and
$\delta_0 = \delta / \left(\binom{m}{\lambda}T^{O(\lambda)}\right)$, where $\lambda = \lambda_{\eps/8, m}(g)$.
\end{enumerate}
\end{lemma}

\paragraph{The Algorithm for Computation Paths.} The algorithm $A'$ simply runs a single instance of the basic algorithm $A$ with a smaller error probability. The outputs it provides to the adversary are  rounded as in the sketch switching technique. 

Specifically, $A'$ runs by emulating $A$ with parameters $\eps/8$ and $\delta_0$. Assuming that the output sequence of the emulated $A$ up to the current time $t$ is $v_0, \ldots, v_t$, it generates $w_t$ in exactly the way described in Lemma \ref{lemma:flip_number_for_sketch_switching}: set $w_0 = v_0$, and for any $i > 0$, if $w_{i-1} \in (1 \pm \eps/2) v_i$ then $w_i = w_{i-1}$, and otherwise $w_i = v_i$. The output provided to the adversary at time $t$ would then be $w_t$.

\begin{proof}
As in the proof of Lemma \ref{lem:sketchswitch}, we may assume the adversary to be deterministic. This means, in particular, that the output sequence we provide to the adversary fully determines its stream of updates $(a_1, \Delta_1), \ldots, (a_m, \Delta_m)$.  
Take $\lambda = \lambda_{\eps/8, m}(g)$. Consider the collection of all possible output sequences (with $\log T$ bits of precision) whose $0$-flip number is at most $\lambda$, and note that the number of such sequences is at most $\binom{m}{\lambda} T^{O(\lambda)}$. 
Each output sequence as above uniquely determines a corresponding stream of updates for the deterministic adversary; let $\mathcal{S}$ be the collection of all such streams.

Pick $\delta_0 = \delta / |\mathcal{S}|$. Taking a union bound, we conclude that with probability $1-\delta$, $A$ (instantiated with parameters $\eps/8$ and $\delta_0$) provides an $\eps/8$-strong $g$-tracking guarantee for all streams in $\mathcal{S}$. The proof follows by applying Lemma \ref{lemma:flip_number_for_sketch_switching} to each stream in $\mathcal{S}$.
%
\end{proof}

\begin{remark} [Bit precision of output]
\label{remark:bit_prec}
For the purposes of this paper, we typically think of the bit precision as $O(\log n)$ (for example, in $F_p$-estimation, there are $\poly(n)$ possible outputs). Since we also generally assume that $m = \poly(n)$, the expression for $\delta_0$ is of the form $\delta_0 = \delta / n^{\Theta(\lambda)}$ in this case. We note that while reducing the bit precision of the output slightly improves the bound on $\delta_0$, this improvement becomes negligible for any streaming algorithm whose dependence in the error probability $\delta$ is logarithmic or better; this covers all situations where we apply Lemma \ref{lem:computationpaths} in this paper.
\end{remark}


\section{$F_p$-Estimation}
\label{sec:Fp}
\label{sec:Fp-estimation}
In this section, we introduce our adversarially robust $F_p$ moment estimation algorithms.
Recall that $F_p$ is given by $\|f\|_p^p = \sum_{i} |f_i|^p$ for $p > 0$. For $p=0$, the $F_0$ moment, or the number of distinct elements, is the number of non-zero coordinates in $f$, that is, $\|f\|_0 = |\{ i \in [n] : f_i \neq 0 \}|$. Recall that in Corollary \ref{cor:fpflip}, we bounded the flip number of the $F_p$ moment in insertion-only streams for any fixed $p > 0$ by $O(\max\{p,1\} \cdot \eps^{-1} \log n)$. By using our sketch switching argument, the strong $F_p$ tracking guarantees of \cite{blasiok2017continuous} as stated in Lemma \ref{lem:fpstrong}, we obtain our first result for $0 < p \leq 2$.

\begin{theorem}[$F_p$-estimation by sketch switching]\label{thm:fpswitch}
Fix any $0 < \eps,\delta \leq 1$ and $0 < p \leq 2$. There is a streaming algorithm for the insertion-only adversarial model which, with probability $1-\delta$, successfully outputs at each step $t \in [m]$ a value $R^t$ such that $R^t= (1 \pm \eps) \|f^{(t)}\|_p$. The space used by the algorithm is $$O\left( \frac{1}{\eps^3} \log n \log \eps^{-1} (\log \eps^{-1}  + \log \delta^{-1}  + \log \log n) \right).$$
\end{theorem}
\begin{proof}
By an application of Lemma \ref{lem:sketchswitch} along with the flip number bound of Corollary  \ref{cor:fpflip} and the strong tracking algorithm of Lemma \ref{lem:fpstrong}, we immediately obtain a space complexity of $$O\left( \frac{1}{\eps^3} \log^2 n (\log \eps^{-1}  + \log \delta^{-1} + \log \log n) \right).$$ We now describe how the factor of $\frac{1}{\eps}\log n$, coming from running $\lambda_{\eps,m} = \Theta(\frac{1}{\eps}\log n)$ independent sketches in Lemma \ref{lem:sketchswitch}, can be improved to $\frac{1}{\eps}\log \eps^{-1}$.

To see this, we change Algorithm \ref{alg:sample} in the following way. Instead of $\Theta(\frac{1}{\eps}\log n)$ independent sketches, we use $\lambda \leftarrow \Theta(\frac{1}{\eps}\log \eps^{-1} )$ independent sketches, and change line $10$ to state $\rho \leftarrow \rho + 1 \pmod \lambda$. Each time we change $\rho$ to $\rho+1 \pmod \lambda$ and begin using the new sketch $A_{\rho+1 \pmod \lambda}$, we \textit{completely restart} the algorithm $A_\rho$ with new randomness, and run it on the remainder of the stream (or until it is restarted again after looping through all $\lambda$ sketches). The proof of correctness in Lemma \ref{lem:sketchswitch} is completely unchanged, except for the fact that now $A_\rho$ is run only on a \textit{sub-interval} $a_j,a_{j+1},\dots,$ of the stream, starting from the time step $j$ where $A_\rho$ is reinitialized and ending at the next time that $A_{\rho}$ is reinitialized. Specifically, at each time step $t \geq j$, $A_\rho$ will produce a $(1 \pm \eps)$ estimate of $\|f^{(t)} - f^{(j-1)}\|_p$ instead of $\|f^{(t)}\|_p$. However, since the sketch will not be used again until a step $t'$ where $\|f^{(t')} \|_p \geq (1+\eps)^{ \lambda}\|f^{(j)}\|_p =  \frac{100}{\eps} \|f^{(j)}\|_p$, it follows that only an $\eps$ fraction of the $\ell_p$ mass was missed by $A_\rho$. In particular, $\|f^{(t')} - f^{(j-1)}\|_p = (1 \pm \eps/100) \|f^{(t')}\|_p$, and thus by giving a $(1 \pm \eps/10)$ approximation of $\|f^{(t')} - f^{(j-1)}\|_p$, the algorithm $A_\rho$ gives the desired $(1 \pm \eps)$
approximation of the underlying $\ell_p$ norm, which is the desired result after a constant factor rescaling of $\eps$. Note that this argument could be used for the $L_0$ norm, or any $L_p$ norm for $p \geq 0$, using an $F_p$ strong tracking algorithm for the relevant $p$. 
\end{proof}

\begin{remark}[The restart trick]
The above proof improves a $\log n$ factor to a $\log{1/\eps}$ one by maintaining independent copies of the sketch in a cyclic manner, where old copies are restarted with fresh randomness (rather than scrapped entirely). This trick works because the $F_p$-value cannot decrease in insertion-only streams, and turns out useful in many insertion-only streaming problems where one wishes to estimate a non-decreasing quantity; we shall see a few examples throughout the paper. Indeed, as long as the previous estimate (using the old randomness) of a certain copy is only, say, a $\eps/10$-fraction of the current estimate, the restart does not majorly effect the output.
\end{remark}

While for most values of $\delta$, the above theorem has better space complexity than the computation paths reduction, for the regime of very small failure probability it is actually preferable to use the latter, as we now state. 

\begin{theorem}[$F_p$-estimation for small $\delta$]\label{thm:fpcompdelta}
Fix any $0 < \eps < 1$, $0 < p \leq 2$, and $\delta < n^{-C \frac{1}{\eps}\log n}$ for a sufficiently large constant $C>1$. There is a streaming algorithm for the insertion-only adversarial model which, with probability $1-\delta$,  successfully outputs at each step $t \in [m]$ a value $R^t$ such that $R^t= (1 \pm \eps) \|f^{(t)}\|_p$. The required space is $O\left( \frac{1}{\eps^2} \log n \log \delta^{-1} \right)$ bits.
\end{theorem}
The proof is a direct application of Lemma \ref{lem:computationpaths}, along with the flip number bound of Corollary \ref{cor:fpflip}, and the $O(\eps^{-2} \log n \log \delta^{-1})$ static $F_p$ estimation algorithm of \cite{kane2010exact}. Indeed, note that the flip number is $\lambda = O (\log n / \eps)$ and that for small enough values of $\delta$ as in the lemma, one has $\log(m^\lambda / \delta) = \Theta(\log(1/\delta))$.


Next, we show that for turnstile streams with $F_p$ flip number $\lambda$, we can estimate $F_p$ with error probability $\delta = n^{-\lambda}$. The space requirement of the algorithm is optimal for algorithms with such failure probability $\delta$, which follows by an $\Omega(\eps^{-2} \log n \log \delta^{-1})$ lower bound for turnstile algorithms \cite{jayram2013optimal}, where the hard instance in question has small $F_p$ flip number.\footnote{The hard instance in \cite{jayram2013optimal} is a stream where $O(n)$ updates are first inserted and then deleted, thus the flip number is at most twice the $F_p$ flip number of an insertion-only stream.}


\begin{theorem}[$F_p$-estimation for $\lambda$-flip number turnstile streams]\label{thm:fpturn}
Let $\mathcal{S}_\lambda$ be the set of all turnstile streams with $F_p$ flip number at most $\lambda \geq \lambda_{\eps,m}(\|\cdot\|_p^p)$ for any $0 < p \leq 2$. Then there is an adversarially robust streaming algorithm for the class $\mathcal{S}_\lambda$ of streams that, with probability $1-n^{-C\lambda}$ for any constant $C>0$, outputs at each time step a value $R^t$ such that $R^t = (1 \pm \eps)\|f\|_p^p$. The space used by the algorithm is $O(\eps^{-2} \lambda \log^2 n )$.
\end{theorem}
\begin{proof}
The proof follows by simply applying Lemma \ref{lem:computationpaths}, along with the $O(\eps^{-2} \log n \log \delta^{-1})$ bit turnstile algorithm of \cite{kane2010exact}.
\end{proof}

In addition, we show that the $F_p$ moment can also be robustly estimated for $p > 2$. In this case, it is preferable to use our computation paths reduction, because the upper bounds for $F_p$ moment estimation for large $p$ yield efficiency gains when setting $\delta$ to be small.

\begin{theorem}[$F_p$-estimation, $p > 2$, by Computation Paths]\label{thm:F_p_estim_largep}
Fix any $\eps,\delta >0$, and any constant $p > 2$. Then there is a streaming algorithm for the insertion-only adversarial model which, with probability $1-n^{-(c \log n )/\eps}$ for any constant $c>1$, successfully outputs at every step $t \in [m]$ a value $R^t$ such that $R^t= (1 \pm \eps) \|f^{(t)}\|_p$. The space used by the algorithm is $O(n^{1-2/p}  ( \eps^{-3} \log^2 n + \eps^{-6/p} ( \log^2 n )^{2/p} \log n ) )$.
\end{theorem}
\begin{proof}
We use the insertion-only $F_p$ estimation algorithm of \cite{ganguly2018high}, which achieves \[\left(n^{1-2/p} \left( \eps^{-2} \log \delta^{-1} + \eps^{-4/p} \log^{2/p} \delta^{-1} \log n \right) \right)\] bits of space in the turnstile (and therefore insertion-only) model. We can set $\delta = \delta/m$ to union bound over all steps, making it a strong $F_p$ tracking algorithm with $$O\left(n^{1-2/p} \left( \eps^{-2}\log (n \delta^{-1}) + \eps^{-4/p} \log^{2/p} (n\delta^{-1} ) \log n \right) \right)$$ bits of space. Then by Lemma \ref{lem:computationpaths} along with the flip number bound of Corollary \ref{cor:fpflip}, the claimed space complexity follows. 
\end{proof}

\section{Distinct Elements Estimation}
\label{sec:F0}
We now demonstrate how our sketch switching technique can be used to estimate the number of \textit{distinct elements}, also known as $F_0$ estimation, in an adversarial stream. In this case, since there exist  static $F_0$ strong tracking algorithms \cite{blasiok2018optimal} which are more efficient than repeating the sketch $\log \delta^{-1}$ times, it will be preferable to use our sketch switching technique. 

\begin{theorem}[Robust Distinct Elements by Sketch Switching]\label{thm:distinct_elements_sketch_switching}
There is an algorithm which, when run on an adversarial insertion-only stream, produces at each step $t \in [m]$ an estimate $R^t$ such that $R^t = (1 \pm \eps)\|f^{(t)}\|_0$ with probability at least $1-\delta$. The space used by the algorithm is $O(\frac{ \log \eps^{-1}}{\eps} (\frac{\log \eps^{-1} + \log \delta^{-1} + \log \log n}{\eps^2} + \log n) )$ bits.
\end{theorem}

\begin{proof}
We use the insertion-only distinct elements \textit{strong} tracking algorithm of \cite{blasiok2018optimal}. Specifically, the algorithm of
 \cite{blasiok2018optimal} uses space $O(\frac{\log \delta_0^{-1} + \log \log n}{\eps^2} + \log n)$, and with probability $1-\delta_0$, successfully returns an estimate $R^t$ for every step $t \in [m]$ such that $R^t = (1 \pm \eps)\|f^{(t)}\|_0$  in the non-adversarial setting. Then by an application of Lemma \ref{lem:sketchswitch}, along with the flip number bound of $O(\log n/\eps)$ from Corollary \ref{cor:fpflip}, we obtain the space complexity with a factor of $\frac{\log n}{\eps}$ blow-up after setting $\delta_0 = \Theta(\delta \frac{\eps}{\log n})$. This gives a complexity of $O(\frac{ \log n}{\eps} (\frac{\log \eps^{-1} + \log \delta^{-1} + \log \log n}{\eps^2} + \log n) )$. To reduce the extra $\log n$-factor to a $\log \eps^{-1}$ factor, we just apply the same argument used in the proof of Theorem \ref{thm:fpswitch}, which shows that by restarting sketches it suffices to keep only $O(\eps^{-1} \log \eps ^{-1} )$ copies.
 \end{proof}

\subsection{Fast Distinct Elements Estimation}
As noted earlier, there are many reasons why one may prefer one of the reductions from Section \ref{sec:framework} to the other. In this section, we will see such a motivation. Specifically, we show that adversarially robust $L_0$ estimation can be accomplished with extremely fast update time using the computation paths reduction of Lemma \ref{lem:computationpaths}. 

 First note that the standard approach to obtaining failure probability $\delta$ is to repeat the estimation algorithm $\log \delta^{-1}$ times independently, and take the median output. However, this blows up the update time by a factor of $\log \delta^{-1}$. Thus black-box applying Lemma \ref{lem:computationpaths} by setting $\delta$ to be small can result in a larger update time. To improve upon this, we will introduce an insertion-only distinct elements estimation algorithm, with the property that the runtime dependency on $\delta^{-1}$ is very small (roughly $\log^2 \log \delta^{-1}$). Thus applying Lemma \ref{lem:computationpaths} on this algorithm results in a very fast robust streaming algorithm. 

\begin{lemma}\label{lem:fastlp}
There is a streaming algorithm which, with probability $1-\delta$, returns a $(1 \pm \eps)$ multiplicative estimate of the number of distinct elements in an insertion-only data stream. The space required is $O(\frac{1}{\eps^2} \log  n (\log  \log n + \log \delta^{-1}))$,\footnote{We remark that it is possible to optimize the $\log n$ factor to $O(\log \delta^{-1} + \log \eps^{-1} + \log \log n)$ by hashing the identities stored in the lists of the algorithm to a domain of size $\poly(\delta^{-1} ,\eps^{-1}, \log n)$. However, in our application we will be setting $\delta \ll 1/n$, and so the resulting adversarially robust algorithm would actually be \textit{less} space efficient. } and the worst case running time per update is $O\left(\left(\log^2 \log \frac{\log n}{\delta}\right) \cdot \left(\log \log \log  \frac{\log n}{\delta}\right)\right)$.
\end{lemma}
We note that previously, the best known update time for insertion-only distinct elements estimation is the algorithm of \cite{kane2010optimal}, which obtains $O(1)$-update time in $O(\eps^{-2} + \log n)$ space with constant failure probability $\delta$. Thus, to obtain small error probability $\delta$, one would need to repeat the entire algorithm $O(\log \delta^{-1})$ times, causing a blow-up of $O(\log \delta^{-1})$ in the update time. 
Before presenting our proof of Lemma \ref{lem:fastlp}, we state the following proposition which will allow for the fast evaluation of $d$-wise independent hash functions.

\begin{proposition}[\cite{von2013}, Ch. 10]\label{prop:fastmulti}
	Let $R$ be a ring, and let $p \in R[x]$ be a degree $d$ univariate polynomial over $R$. Then given distinct $x_1,x_2,\dots,x_d \in R$, all the values $p(x_1),p(x_2),\dots,p(x_d)$ can be computed using $O(d \log^2 d \log \log d)$ operations over $R$.
\end{proposition}

\begin{proof}[Proof of Lemma \ref{lem:fastlp}]
	We describe the algorithm here, as stated in Algorithm \ref{alg:f0fast}. 
	
	\begin{algorithm}[ht]
		Initialize Lists $L_0,L_1,\dots,L_\ell \leftarrow \emptyset$, for $\ell$ chosen such that $n^2 \leq 2^\ell \leq n^3$.
		$B \leftarrow \Theta(\frac{1}{\eps^2}( \log \log n + \log \delta^{-1}))$, $\; d \leftarrow \Theta(\log \log n + \log \delta^{-1})$\\
		Initialize $d$-wise independent hash function $H: [n] \to [2^\ell]$. \\
		\While{Receive update $a_i \in [n]$}{
			Let $j$ be such that $ 2^{\ell - j - 1} \leq H(a_i) < 2^{\ell - j}$ \\
			\If{$L_j$ has not been deleted}{
				Add $a_i$ to the list $L_j$ if it is not already present.\\
			}
			If $|L_j| > B$ for any $j$, delete the list $L_j$, and never add any items to it again. \\
		}
		Let $i$ be the largest index such that $|L_i| \geq \frac{1}{5}B$.\\
		Return $2^{i+1}|L_i|$ as the estimate of $\|f\|_0$\\
		\caption{Fast non-adversarial distinct elements estimation. }
		\label{alg:f0fast}
	\end{algorithm}
	
	We initialize lists $L_0,L_1,\dots,L_\ell \leftarrow \emptyset$, where $\ell$ is set so that $n^2 \leq 2^\ell \leq  n^3$. We also choose a hash function $H: [n] \to [2^\ell]$. The lists $L_i$ will store a set of identities $L_i \subset [n]$ which have occurred in the stream. We also set $B \leftarrow \Theta(\frac{1}{\eps^2}( \log \log n + \log \delta^{-1}))$. For now, assume that $H$ is fully independent.
	
	At each step when we see an update $a_i \in [n]$ (corresponding to an update which increments the value of $f_i$ by one), we compute $j$ such that $2^{\ell - j -1} \leq H(a_i) \leq 2^{\ell - j}$. Note that this event occurs with probability $2^{-(j+1)}$. Then we add the $O(\log n)$-bit identity $a_i$ to the list $L_j$ if $|L_j| < B$. Once $|L_k| = B$ for any $k \in [\ell]$, we delete the entire list $L_k$, and never add an item to $L_k$ again. We call such a list $L_k$ \textit{saturated}. At the end of the stream, we find the largest value $i$ such that $\frac{1}{5}B \leq |L_i|$, and output $2^{i+1} |L_i|$ as our estimate of $\|f\|_0$. 
	
	We now analyze the above algorithm. Let $i_0$ be the smallest index such that $\ex{|L_{i_0}| } \leq \|f\|_0 2^{-(i_0+1)} <\frac{1}{5(1+\eps)}B$. Note here that $\ex{|L_k|} = 2^{-(k+1)}\|f\|_0$ for any $k \in [\ell]$. By a Chernoff bound, with probability $1-\exp(-\Omega(-\eps^2 B)) < 1 - \delta^2/\log(n)$ we have that $|L_{i_0}| < \frac{1}{5}B$. We can then union bound over all such indices $i \geq  i_0$. This means that we will not output the estimate used from any index $i \geq i_0$. Similarly, by a Chernoff bound we have that $|L_{i_0 - 1}| = (1 \pm \eps) \|f\|_0 2^{-i_0 } < \frac{2}{5} B$ and $|L_{i_0 - 2}| = (1 \pm \eps) \|f\|_0 2^{-i_0 + 1}$, and moreover we have $\frac{2}{5(1+\eps)}B \leq \|f\|_0 2^{-i_0 + 1}  \leq \frac{4}{5}B$, meaning that the output of our algorithm will be either $|L_{i_0 - 1}|2^{i_0}$ or $|L_{i_0 - 2}|2^{i_0 -1}$, each of which yields a $(1 \pm \eps)$ estimate. Now note that we cannot store a fully independent hash function $H$, but since we only needed all events to hold with probability $1-\Theta(\delta^2/\log(n))$, it suffices to choose $H$ to be a $d$-wise independent hash function for $d = O(\log \log n + \log \delta^{-1})$, which yields Chernoff-style tail inequalities with a decay rate of $\exp(-\Omega(d))$ (see e.g. Theorem 5 of \cite{schmidt1995chernoff}).

	Next we analyze the space bound. Trivially, we store at most $O(\log n)$ lists $L_i$, each of which stores at most $B$ identities which require $O(\log n)$ bits each to store, yielding a total complexity of $O( \frac{1}{\eps^2} \log^2 n (\log \log n + \log \delta^{-1} ))$. We now show however that at any given step, there are at most $O(B )$ many identities stored in all of the active lists. To see this, let $i_0 <i_1< \dots < i_s$ be the time steps such that $\|f^{(i_j)}\|_0 = 2^{j+1}\cdot B$, and note that $s \leq \log(n)+1$. Note that before time $i_0$, at most $B$ identities are stored in the union of the lists. First, on time step $i_j$ for any $j \in [s]$, the expected size of $|L_{j - 2}|$ is at least $2|B|$ (had we never deleted saturated lists), and, with probability $1-(\delta/\log n)^{10}$ after a union bound, it holds that $|L_{j'}|$ is saturated for all $j' \leq j -2$. Moreover, note that the expected number of identities written to lists $L_{j'}$ with $j' \geq j-1$ is $\|f^{(i_j)}\|_0 \sum_{\nu \geq 1} 2^{-j+1+\nu} \leq 2B$, and is at most $4B$ with probability at least $1-(\delta/\log n)^{10}$ (using the $d$-wise independence of $H$). We conclude that on time step $i_j$, the total space being used is $O(B \log n)$ with probability at least $1-(\delta/\log n)^{10}$, so we can union bound to obtain that this space holds over all such steps $i_j$ for $j \in [s]$.
	
	Next, we must analyze the space usage at steps $\tau$ for $i_j < \tau < i_{j+1}$. Note that the number of new distinct items which occur over all such time steps $\tau$ is at most $2^{j+1} \cdot B$ by definition. Since we already conditioned on the fact that $|L_{j'}|$ is saturated for all $j' \leq j-2$, it follows that each new item is written into a list with probability at most $2^{-j}$. Thus the expected number of items which are written into lists within times $\tau$ satisfying $i_j < \tau < i_{j+1}$ is $2^{j+1} \cdot B \cdot 2^{-j} = 2 B$ in expectation, and at most $8B$ with probability $1- (\delta/\log n)^{10}$ (again using the $d$-wise independence of $H$). Conditioned on this, the total space used in these steps is at most $O(B \log n ) = O(\frac{1}{\eps^2} \log n (\log \log n + \log \delta))$ in this interval, and we then can union bound over all such $O(\log n)$ intervals, which yields the desired space.

	Finally, for the update time, note that at each stream update $a_i \in [n]$, on the first step of the algorithm, we compute the value of a $d$-wise independent hash function $H$. Na\"ively, computing a $d$-wise independent hash function requires $O(d)$ arithmetic operations (in the standard RAM model), because $H$ in this case is just a polynomial of degree $d$ over $\Z$. On the other hand, we can batch sequences of $d = O(\log \log n + \log \delta^{-1})$ computations together, which require an additive $O(d\log n) = O(\log n(\log \log n + \log \delta^{-1}))$ bits of space at any given time step to store (which is dominated by the prior space complexity). Then by Proposition \ref{prop:fastmulti}, all $d$ hash function evaluations can be carried out in $O(d \log^2 d \log \log d) = O(d \log^2 (\log \frac{\log n }{\delta}) \log \log \log \frac{\log n}{\delta})$ time. The work can then be evenly distributed over the following $d$ steps, giving a worst case update time of   $O(\log^2 (\log \frac{\log n }{\delta}) \log \log \log \frac{\log n}{\delta})$. Note that this delays the reporting of the algorithm for the contribution of updates by a total of $d$ steps, causing an additive $d$ error. However, this is only an issue if $d \geq \eps \|f\|_0$, which occurs only when $\|f\|_0 \geq \frac{1}{\eps} d$. Thus for the first $D =O( \eps^{-1} d)$ distinct items, we can store the non-zero items exactly (and deterministically), and use the output of this deterministic algorithm. The space required for this is $O(\eps^{-1} \log(n) (\log \log n + \log \delta^{-1})$, which is dominated by the space usage of the algorithm overall. After $D$ distinct items have been seen, we switch over to using the output of the randomized algorithm described here. Finally, the only other operation involves adding an identity to at most one list per update, which is $O(1)$ time, which completes the proof.
\end{proof}

We can use the prior result of Lemma \ref{lem:fastlp}, along with our argument for union bounding over adversarial computation paths of Lemma \ref{lem:computationpaths} and the flip number bound of Corollary \ref{cor:fpflip}, which results in an adversarially robust streaming algorithm for distinct elements estimation with extremely fast update time. 

\begin{theorem}\label{thm:fast}
There is a streaming algorithm which, with probability $1-n^{-(C/\eps) \log n}$ for any constant $C \geq 1$,  when run on an adversarially chosen insertion-only data stream, returns a $(1 \pm \eps)$ multiplicative estimate of the number of distinct elements at every step in the stream. The space required is $O(\frac{1}{\eps^3} \log^3  n )$, and the worst case running time is $O\left( \left(\log^2 \frac{\log n}{\eps} \right) \cdot \left(\log \log \frac{\log n}{\eps} \right)\right)$ per update.
\end{theorem}

\section{Heavy Hitters}
\label{sec:HH}

In this section, we study the popular \textit{heavy-hitters} problem in data streams. The heavy hitters problem tasks the algorithm with recovering the most frequent items in a data-set. Stated simply, the goal is to report a list $S$ of items $f_i$ that appear least $\tau$ times, meaning $f_i \geq \tau$, for a given threshold $\tau$. Generally, $\tau$ is parameterized in terms of the $L_p$ norm of the frequency vector $f$, so that $\tau = \eps \|f\|_p$. For $p>2$, this problem is known to take polynomial space \cite{AlonMS96,bar2004information}. Thus, the strongest such guarantee that can be given in sub-polynomial space is known as the $L_2$ guarantee: 
\begin{definition}\label{def:HH}
A streaming algorithm is said to solve the $(\eps,\delta)$-heavy hitters problem with the $L_2$ guarantee if the algorithm, when run on a stream with frequency vector $f \in \R^n$,  outputs a set $S \subset [n]$ such that with probability $1-\delta$ the following holds: for every $i \in [n]$ if $|f_i| \geq \eps \|f\|_2$ then $i \in S$, and if $|f_i| \leq (\eps/2) \|f\|_2$ then $i \notin S$. 
\end{definition}
\noindent
 We also introduce the related task of $(\eps,\delta)$-point queries.
\begin{definition}
A streaming algorithm is said to solve the $(\eps,\delta)$ point query problem with the $L_2$ guarantee if with probability $1-\delta$, at every time step $t \in [m]$, for each coordinate $i \in [n]$ it can output an estimate $\wh{f}^t_i$ such that $|\wh{f}^t_i - f^{(t)}_i| \leq \eps \|f^{(t)}\|_2$. Equivalently, it outputs a vector $\widehat{f}^t \in \R^n$ such that $\|f^{(t)} - \widehat{f}^t\|_\infty \leq \eps \|f^{(t)}\|_2$.\footnote{We note that a stronger form of error is possible, called the \textit{tail guarantee}, which does not count the contribution of the top $1/\eps^2$ largest coordinates to the error $\eps \|f\|_2$. We restrict to the simpler version of the $L_2$ guarantee.}
\end{definition}

 Notice that for any algorithm that solves the $(\eps,\delta)$-point query problem, if it also has estimates $R^t = (1 \pm \eps/10)\|f^{(t)}\|_2$ at each time step $t \in [m]$, then it immediately gives a solution to the $(\eps,\delta)$-heavy hitters problem by just outputting all $i \in [n]$ with $\tilde{f}^t_i > (3/4)\eps R^t$. Thus solving $(\eps,\delta)$-point queries, together with $F_2$ tracking, is a stronger property. In the following, we say that $\widehat{f}^t$ is $\eps$-correct at time $t$ if $\|f^{(t)} - \widehat{f}^t\|_\infty \leq \eps \|f^{(t)}\|_2$.
 
In this section, we demonstrate how this fundamental task of point query estimation can be accomplished robustly in the adversarial setting.  Note that we have already shown how $F_2$ tracking can be accomplished in the adversarial model, so our focus will be on point queries. Our algorithm relies on a similar sketch switching technique as used in Lemma \ref{lem:sketchswitch}, which systematically hides randomness from the adversary by only publishing a new estimate $\widehat{f}^t$ when absolutely necessary. To define what is meant by ``absolutely necessary", we will first need the following proposition.

\begin{proposition}\label{prop:HH}
Suppose that $\wh{f}^t \in \R^n$ is $\eps$-correct at time $t$ on an insertion-only stream, and let $t_1 > t$ be any time step such that $\|f^{(t_1)} - f^{(t)}\|_\infty \leq \eps \|f^{(t)}\|_2$. Then $\wh{f}^t$ is $2\eps$-correct at time $t_1$.
\end{proposition}
\begin{proof}
    $
        \|\wh{f}^{t} - f^{(t_1)}\|_\infty \leq \|\wh{f}^{t} - f^{(t)}\|_\infty + \|f^{(t_1)} - f^{(t)}\|_\infty \leq \eps \|f^{(t)}\|_2 + \eps \| f^{(t)}\|_2  \leq 2\eps \|f^{(t_1)}\|_2$. 
\end{proof}


To prove the main theorem of Section \ref{sec:HH}, we will need the classic \textit{count-sketch} algorithm for finding $L_2$ heavy hitters of Charikar et al.~\cite{charikar2004finding}, which solves the more general point query problem in the static setting with high probability. 

\begin{lemma}[\cite{charikar2004finding}]\label{lem:countsketch}
There is a streaming algorithm in the non-adversarial insertion-only model which solves the $(\eps,\delta)$-point query problem, using  $O(\frac{1}{\eps^2}\log n \log \frac{n}{\delta} )$ bits of space. 
\end{lemma}

We are now ready to prove the main theorem of this section.

\begin{theorem}[$L_2$ point query and heavy hitters]\label{thm:HH}
Fix any $\eps,\delta>0$. 
	There is a streaming algorithm in the adversarial insertion-only model which solves the $(\eps, n^{-C})$ point query problem, and  also the $O(\eps,n^{-C})$-heavy hitters problem, for any constant $C>1$. The algorithm uses $O(\frac{\log \eps^{-1}}{\eps^3}\log^2 n)$ bits of space. 
\end{theorem}

\begin{proof}
Since we already know how to obtain estimates $R^t = (1 \pm \eps/100)\|f^{(t)}\|_2$ at each time step $t \in [m]$ in the adversarial insertion-only model within the required space, it will suffice to show that we can obtain estimates $\wh{f}^t$ which are $\eps$-correct at each time step $t$ (i.e., it will suffice to solve the point query problem).

	Let $1=t_1 <t_2 <\dots <t_T=m$ for $T =\Theta( \eps^{-1} \log n)$ be any set of time steps such that $\|f^{(t_{i+1})} -f^{(t_{i})} \|_2 \leq \eps \|f^{(t_{i})}\|_2$ for each $i \in [T-1]$.
	Then by Proposition \ref{prop:HH}, using that $\|f^{(t_{i+1})} -f^{(t_{i})} \|_\infty \leq \|f^{(t_{i+1})} -f^{(t_{i})} \|_2$, we know that if we output an estimate $\wh{f}^i$ which is $\eps$-correct for time $t_i$, then $\wh{f}^i$ will still be $2\eps$ correct at time $t_{i+1}$. Moreover, because the stream is insertion-only, the frequency vectors $f^{(t)}$ are coordinate-wise monotonically increasing over time. The latter implies that $\|f^{(t)} -f^{(t_{i})} \|_2 \leq \|f^{(t_{i+1})} -f^{(t_{i})} \|_2 $ for all $t \in [t_i, t_{i+1}]$, and therefore if $\wh{f}^i$ is $\eps$-correct for time $t_i$, then $\wh{f}^i$ will also be $2\eps$ correct at any time $t \in [t_i, t_{i+1}]$. Thus our approach will be to output vectors $\wh{f}^1,\wh{f}^2,\dots,\wh{f}^T$, such that we output the estimate $\wh{f}^i \in \R^n$ at all times $\tau$ such that $t_i \leq \tau < t_{i+1}$, and such that $\wh{f}^i$ is  $\eps$-correct for time $t_i$. 
	
	First, to find the time steps $t_i$, we run the adversarially robust $F_2$ estimator of Theorem \ref{thm:fpswitch}, which gives an estimate $R^t = (1 \pm \eps/100)\|f^{(t)}\|_2$ at each time step $t \in [m]$ with probability $1-n^{-C}$ for any constant $C>1$, and uses space $O(\eps^{-3} \log^2{n} \log \eps^{-1})$. Notice that this also gives the required estimates $R^t$ as stated above. By rounding down the outputs $R^t$ of this $F_2$ estimation algorithm to the nearest power of $(1+\eps/2)$, we obtain our desired points $t_i$. Notice that this also gives $T =\Theta( \eps^{-1} \log n)$ as needed, by the flip number bound of Corollary \ref{cor:fpflip}.
	
	Next, to obtain the desired $\eps$ point query estimators at each time step $t_i$, we run $T$ independent copies of the point query estimation algorithm of Lemma \ref{lem:countsketch}. At time $t_i$, we use the output vector of the $i$-th copy as our estimate $\wh{f}^i$, which will also be used without any modification on all times $\tau$ with $t_i \leq \tau < t_{i+1}$. Since each copy of the algorithm only reveals any of its randomness at time $t_i$, at which point it is never used again, by the same argument as in Lemma  \ref{lem:sketchswitch} it follows that each $\wh{f}^i$ will be $\eps$-correct for time $t_i$. Namely, since the set of stream updates on times $1,2,\dots,t_i$ are independent of the randomness used in the $i$-th copy of point-estimation algorithm, we can deterministically fix the updates on these time steps, and condition on the $i$-th copy of the non-adversarial streaming algorithm being correct on these updates. Therefore this algorithm correctly solves the $2\eps$ point query problem on an adversarial stream. The total space used is 
\[	O\left(	\eps^{-3} \log^2{n} \log \eps^{-1} + T\eps^{-2} \log^2{n} \right)
~.
\]
We now note that we can improve the space by instead running only $T' = O(\eps^{-1} \log \eps^{-1})$ independent copies of the algorithm of Lemma \ref{lem:countsketch}. Each time we use one of the copies to output the desired estimate $\wh{f}^i$, we completely restart that algorithm on the remaining suffix of the stream, and we loop modularly through all $T'$ copies of the algorithm, at each step using the copy that was least recently restarted to output an estimate vector. More formally, we keep copies $\mathcal{A}_1,\dots,\mathcal{A}_{T'}$ of the algorithm of Lemma \ref{lem:countsketch}. Each time we arrive at a new step $t_i$ and must produce a new estimate $\wh{f}^i$, we query the algorithm $\mathcal{A}_j$ that was \textit{least recently restarted}, and use the estimate obtained by that algorithm, along with the estimates $R^t$. 

  The same correctness argument will hold as given above, except now each algorithm, when used after being restarted at least once, will only be $\eps$-correct for the frequency vector defined by a sub-interval of the stream. However, by the same argument used in Theorem \ref{thm:fpswitch}, we can safely disregard the prefix that was missed by this copy of the algorithm, because it contains only an $\eps/100$-fraction of the total $L_p$ mass of the current frequency vector when it is applied again. Formally, if an algorithm $\mathcal{A}_j$ is used again at time $t_i$, and it was last restarted at time $\tau$, then by the correctness of our estimates $R^t$, the $L_2$ norm must have gone up by a factor of $(1+\eps)^{T'} = \frac{100}{\eps}$, so $\|f^{(\tau)}\|_2 \leq \eps /100\|f^{(t_i)}\|_2$. Moreover, we have that the estimate $\wh{f}^i$ produced by the algorithm $\mathcal{A}_j$ at time $t_i$ satisfies 
  $\|\wh{f}^i - (f^{(t_i)} - f^{(\tau)})\|_\infty \leq \eps \|f^{(t_i)} - f^{(\tau)}\|_2$. This follows from the fact that $(f^{(t_i)} - f^{(\tau)})$ is the frequency vector of the sub-stream on which the algorithm $\mathcal{A}_j$ has been run at time $t_i$, along with the $\eps$-correctness guarantee of the algorithm of Lemma \ref{lem:countsketch}. But then
  \begin{equation}
      \begin{split}
          \|\wh{f}^i - f^{(t_i)} \|_\infty & \leq\|\wh{f}^i - (f^{(t_i)} - f^{(\tau)})\|_\infty + \|f^{(\tau)}\|_\infty \\
     &\leq \eps \|f^{(t_i)} - f^{(\tau)}\|_2 + \|f^{(\tau)}\|_2 \qquad  \\ 
    &\leq \eps \left(\|f^{(t_i)}\|_2 + \| f^{(\tau)}\|_2\right)+ \eps/100 \|f^{(t_i)}\|_2\\ 
        &\leq \eps \|f^{(t_i)}\|_2(1+\eps)+ \eps/100 \|f^{(t_i)}\|_2 \\ 
         &\leq 2\eps \|f^{(t_i)}\|_2, 
      \end{split}
  \end{equation}
 where in the first line we added and subtracted $f^{(\tau)}$ and applied the triangle inequality, in the second line we used the fact that  $\|\wh{f}^i - (f^{(t_i)} - f^{(\tau)})\|_\infty \leq \eps \|f^{(t_i) }- f^{(\tau)}\|_2$ along with the fact that the $\ell_\infty$ norm is bounded by the $\ell_2$ norm, and in the third line we used the triangle inequality.   Thus $\wh{f}^i$ is still $2\eps$-correct at time $t_i$ for the full stream vector $f^{(t_i)}$. So by the same argument as above using Proposition \ref{prop:HH}, it follows that the output of the overall algorithm is always $4\eps$-correct for all time steps $\tau \in [m]$, and we can then re-scale $\eps$ by a factor of $1/4$. 
  Substituting the new number $T'$ of copies used into the above equation, we obtain the desired complexity.  
\end{proof}

\section{Entropy Estimation}
\label{sec:Entropy}

We now show how our general techniques developed in Section \ref{sec:framework} can be used to approximate the empirical Shannon entropy $H(f)$ of an adversarial stream. Recall that for a non-zero vector $f$, we have that $H(f) = - \sum_{i, f_i \neq 0} p_i \log( p_i)$, where $p_i = \frac{|f_i|}{\|f\|_1}$.  Also recall that for $\alpha > 0$, the $\alpha$-Renyi Entropy $H_\alpha(x)$ of $x$ is given by $H_\alpha(x) = \log\left( \frac{\|x\|_\alpha^\alpha}{\|x\|_1^\alpha}\right)/(1-\alpha)$. 

We begin with the following observation, which will allow us to consider multiplicative approximation of $2^{H(x)}$. Then, by carefully bounding the flip number of the Renyi entropy $H_\alpha$ for $\alpha$ close to $1$, we will be able to bound the flip number of $H$.

\begin{remark}
Note that any algorithm that gives an $\eps$-additive approximation of the Shannon Entropy $H(x):\R^n \to \R$ gives a $(1 \pm \eps)$ multiplicative approximation of $g(x) = 2^{H(x)}$, and vice-versa.
\end{remark}

\begin{proposition}[Theorem 3.1 of \cite{harvey2008sketching}]\label{prop:H}
Let $x \in \R^n$ be a probability distribution whose smallest non-zero value is at least $\frac{1}{m}$, where $m \geq n$. Let $0 < \eps < 1$ be arbitrary. Define $\mu = \eps /(4 \log m)$ and $\nu = \eps/(4 \log n \log m)$, $\alpha=1 + \mu/(16 \log(1/\mu))$ and $\beta = 1 + \nu/(16 \log(1/\nu))$. Then
\[ 1 \leq \frac{H_\alpha}{H} \leq 1+\eps  \text{     and    } 0 \leq H - H_\beta \leq   \eps. \]
\end{proposition}

\begin{proposition}
\label{prop:entropyflip}
Let $g: \R^N \to R$ be $g(x) = 2^{H(x)}$, i.e., the exponential of the Shannon entropy. Then the $(\eps,m)$-flip number of $g$ for the insertion-only streaming model is $\lambda_{\eps,m}(g) = O(\frac{1}{\eps^2}\log^3 m (\log \log n + \log \eps^{-1}))$.
\end{proposition}

The proof of the above proposition is given later in this section. We now state the main result on adversarially robust entropy estimation. An improved result is stated for the \emph{random oracle model} in streaming, which means that the algorithm is given random (read-only) access to an arbitrarily large string of random bits.

\begin{theorem}[Robust Additive Entropy Estimation] \label{thm:entropy}
There is an algorithm for $\eps$-additive approximation of entropy in the insertion-only adversarial streaming model which requires a total of $O(\frac{1}{\eps^4} \log^4 n(\log \log n + \log \eps^{-1}))$ bits of space assuming the random oracle model, and $O(\frac{1}{\eps^4} \log^6 n(\log \log n + \log \eps^{-1}))$ bits of space in the general insertion-only model. 
\end{theorem}

To obtain our entropy estimation algorithm of Theorem \ref{thm:entropy}, we will first need to state the results for the state of the art non-adversarial streaming algorithms for additive entropy estimation. The first algorithm is a $O(\eps^{-2} \log^2 n)$-bit streaming algorithm for additive approximation of the entropy of a turnstile stream, which in particular holds for insertion-only streams. The second result is a $\tilde{O}(1/\eps^2)$ upper bound for entropy estimation in the insertion-only model when a random oracle is given.

\begin{lemma}[\cite{clifford2013simple}]\label{lem:clifford}
	There is an algorithm in the turnstile model that gives an $\eps$-additive approximation to the Shannon Entropy $H(f)$ of the stream. The failure probability is $\delta$, and the space required is $O(\frac{1}{\eps^2}\log^2 n \log \delta^{-1})$ bits.
\end{lemma}

\begin{lemma}[\cite{jayaram2019towards}]\label{lem:jw19}
	There is an algorithm in the insertion-only random oracle model that gives an $\eps$-additive approximation to the Shannon Entropy $H(f)$ of the stream. The failure probability is $\delta$, and the space required is $O(\frac{1}{\eps^2} (\log \delta^{-1}+ \log \log n + \log \eps^{-1}))$
\end{lemma}

We now give the proof of Proposition \ref{prop:entropyflip}, and then the proof of Theorem \ref{thm:entropy}. 

\begin{proof}[Proof of Proposition \ref{prop:entropyflip}]
	By Proposition $\ref{prop:H}$, it suffices to get a bound on the flip number of $H_\beta$ for the parameters $\beta = 1 + \nu/(16 \log(1/\nu))$ and $\nu = \eps/(4 \log n \log m)$. 
	Recall $g(x) = 2^{H_\beta(x)} = (\|x\|_\beta^\beta/\|x\|_1^\beta)^{1/(1-\beta)} = (\|x\|_1 / \|x\|_\beta)^{\frac{\beta}{\beta-1}} $, and define 
	\[\tau = \eps \cdot \frac{\beta-1}{\beta} = \Theta\left(\frac{\eps^2}{(\log^2 n )\cdot \left( \log \log n + \log \eps^{-1}\right)}\right).\]
	Then, to increase  $g(x)$ by a factor of $(1+\eps)$, one must increase $\|x\|_1 / \|x\|_\beta$ by a factor of $1 + \Omega(\tau)$. Since the stream is insertion-only, both $\|x\|_1$ and $\|x\|_\beta$ are non-decreasing in the stream. Therefore, for the ratio to increase by a factor of $1 + \Omega(\tau)$, it must be that $\|x\|_1$ itself increases by a factor of at least $1 + \Omega(\tau)$. 
		Similarly, for $g(x)$ to decrease by a factor of $1+\eps$, this would requires $\|x\|_\beta$ to increase by a factor of $1 +\Omega(\tau)$.
	
	In summary, if for time steps $1\leq t_1 < t_2\leq m$ of the stream we have $g(f^{(t_2)}) > (1+\eps)g(f^{(t_1)})$, then it must be the case that $\|f^{(t_2)}\|_1 > (1 + \Omega(\tau) )\|f^{(t_1)}\|_1 $. Similarly, if we had  $g(f^{(t_2)}) < (1-\eps)g(f^{(t_1)})$, then it must be the case that $\|f^{(t_2)}\|_\beta > (1 + \Omega(\tau) )\|f^{(t_1)}\|_\beta $. Since $\|f^{(m)}\|_\beta \leq \|f^{(m)}\|_1 \leq Mn$ and $\|\cdot \|_1, \|\cdot\|_\beta$ are monotone for insertion-only streams, it follows that each of them can increase by a factor of $(1+ \Omega(\tau))$ at most $O(\frac{1}{\tau} \log n) = O(\frac{\log^3 n (\log \log n + \log \eps^{-1}) }{\eps^2})$ times during the stream, which completes the proof since $\log n = \Theta(\log m)$.
\end{proof}

\begin{proof}[Proof of Theorem \ref{thm:entropy}]
	The proof follows directly from an application of Lemma \ref{lem:sketchswitch}, using the non-adversarial algorithms of Lemmas \ref{lem:clifford} and \ref{lem:jw19}, as well as the flip number bound of Lemma \ref{prop:entropyflip}. Note that to turn the algorithms of Lemmas \ref{lem:clifford} and \ref{lem:jw19} into tracking algorithms, one must set $\delta < 1/m$, which yields the stated complexity.
\end{proof}

\section{Bounded Deletion Streams}
\label{sec:BoundedDel}
In this section, we show how our results can be used to obtain adversarially robust streaming algorithms for the \textit{bounded-deletion model}, introduced in \cite{jayaram2018data}. The bounded deletion model serves as an intermediate model between the turnstile and insertion-only model. Motivated by common lower bounds for turnstile streams, which utilize seemingly unrealistic hard instances that insert a large number of items before deleting nearly all of them, bounded deletion streams are possibly a more representative model for real-world data streams. Intuitively, a bounded deletion stream is one where the $F_p$ moment of the stream is a $\frac{1}{\alpha}$ fraction of what the $F_p$ moment would have been had all updates been replaced with their absolute values,  meaning that the stream does not delete off an arbitrary amount of the $F_p$ weight that it adds over the course of the stream. Formally, the model is as follows.

\begin{definition}
Fix any $p \geq 1$ and $\alpha \geq 1$. 
	A data stream $u_1,\dots,u_m$, where $u_i = (a_i, \Delta_i) \in [n] \times \{1,-1\}$ are the updates to the frequency vector $f$, is said to be an $F_p$ $\alpha$-bounded deletion stream if at every time step $t \in [m]$  we have 
$	\|f^{(t)}\|_p^p \geq \frac{1}{\alpha} \sum_{i=1}^n (\sum_{t' \leq t : a_{t'} = i} |\Delta_{t'}|)^p$.
\end{definition}
Specifically, the $\alpha$-bounded deletion property says that the $F_p$ moment $\|f^{(t)}\|_p^p$ of the stream is at least  $\frac{1}{\alpha} \|h^{(t)}\|_p^p$, where $h$ is the frequency vector of the stream with updates $u_i' = (a_i ,\Delta_i')$ where $\Delta_i' = |\Delta_i|$ (i.e., the absolute value stream). 
Note here that the model assumes unit updates, i.e., we have $|\Delta_i| = 1$ for each $i \in [m]$, which can be accomplished without loss of generality with respect to the space complexity of algorithms, by simply duplicating integral updates into unit updates. 

In \cite{jayaram2018data}, the authors show that for $\alpha$-bounded deletion streams, a factor of $\log n$ in the space complexity of turnstile algorithms can be replaced with a factor of $\log \alpha$ for many important streaming problems. In this section, we show another useful property of bounded-deletion streams: norms in such streams have bounded flip number. We use this fact to design adversarially robust streaming algorithms for data streams with bounded deletions.

\begin{lemma}\label{lem:flipbd}
Fix any $p \geq 1$.
	The $\lambda_{\eps,m}(\| \cdot \|_p)$ flip number of the $L_p$ norm of a $\alpha$-bounded deletion stream is at most $O(p\frac{\alpha}{\eps^p}\log n)$.
\end{lemma}

\begin{proof} Let $h$ be the frequency vector of the stream with updates $u_i' = (a_i ,\Delta_i')$ where $\Delta_i' = |\Delta_i|$. Note that $h$ is then the frequency vector of an insertion-only stream.
	Now let $0 \leq t_1< t_2<\dots < t_k \leq m$ be any set of time steps such that $\|f^{(t_i)}\|_p \notin (1 \pm \eps) \|f^{(t_{i+1})}\|_p$ for each $i \in [k-1]$. Since by definition of the $\alpha$-bounded deletion property, we have $\|f^{(t)}\|_p \geq \frac{1}{\alpha^{1/p}}\|h^{(t)}\|_p$ for each $t \in [m]$, it follows that 
	\begin{equation}
	\begin{split}
	\| f^{(t_{i+1})} - f^{(t_i)}\|_p &\geq  \left| \|f^{(t_{i+1})}\|_p - \| f^{(t_i)}\|_p \right| 
	\geq \eps\|f^{(t_{i+1})}\|_p  
	\geq \frac{\eps}{\alpha^{1/p}}\|h^{(t_{i+1})}\|_p 
	\geq \frac{\eps}{\alpha^{1/p}}\|h^{(t_{i})}\|_p 
	\end{split}
	\end{equation}
	where in the last inequality we used the fact that $h$ is an insertion-only stream. Now since the updates to $h$ are the absolute value of the updates to $f$, we also have that $\|  h^{(t_{i+1})} - h^{(t_i)}\|_p^p \geq \|  f^{(t_{i+1})} - f^{(t_i)}\|_p^p \geq \frac{\eps^p}{\alpha}\|h^{(t_{i})}\|_p^p$. Thus 
	\begin{equation}
	\begin{split}
	\|h^{(t_{i+1})}\|_p^p & =   \|h^{(t_i)} +  \left( h^{(t_{i+1})} - h^{(t_i)} \right)\|_p^p   
	\geq  \|h^{(t_i)}\|_p^p + \| h^{(t_{i+1})} - h^{(t_i)} \|_p^p   
	\geq  (1+\frac{\eps^p}{\alpha})\|h^{(t_i)}\|_p^p
	\end{split}
	\end{equation}
	where in the second inequality, we used the fact that $\|X + Y\|_p^p \geq \|X\|_p^p + \|Y\|_p^p$ for non-negative integral vectors $X,Y$ when $p \geq 1$. Thus $\|h^{(t_{i+1})}\|_p^p$ must increase by a factor of $(1 + \eps^p/\alpha)$ from $\|h^{(t_{i})}\|_p^p$ whenever $\|f^{(t_i)}\|_p \notin (1 \pm \eps) \|f^{(t_{i+1})}\|_p$. Since $\| 0 \|_p^p = 0$, and $\|h^{(m)}\|_p^p \leq M^p n \leq n^{cp}$ for some constant $c>0$, it follows that this can occur at most $O(p\frac{\alpha}{\eps^p}\log n)$ many times. Thus $k = O(p\frac{\alpha}{\eps^p}\log n)$, which completes the proof. 
\end{proof}

We now use our computation paths technique of Lemma \ref{lem:computationpaths}, along with the space optimal turnstile $F_p$ estimation algorithm of \cite{kane2010exact}, to obtain adversarially robust algorithms for $\alpha$-bounded deletion streams. Specifically, we show that we can estimate the $F_p$ moment of a bounded deletion stream robustly. We remark that once $F_2$ moment estimation can be done, one can similarly solve the heavy hitters problem in the robust model using a similar argument as in Section \ref{sec:HH}, except without the optimization used within the proof of Theorem \ref{thm:HH} which restarts sketches on a suffix of the stream. The resulting space would be precisely an $(\frac{\alpha}{\eps}\log n)$-factor larger than the space stated in Theorem \ref{thm:HH}. 


\begin{theorem}\label{thm:bd}
Fix $p \in [1,2]$, $\alpha \geq 1$, and any constant $C>1$. Then there is an adversarially robust $F_p$ estimation algorithm for $\alpha$-bounded deletion streams which, with probability $1-n^{-C}$, returns at each time step $t \in [m]$ an estimate $R^t$ such that $R^t = (1 \pm \eps) \|f^{(t)}\|_p^p$. The space used by the algorithm is $O( \alpha \eps^{-(2+p)} \log^3 n)$.
\end{theorem}   

\begin{proof}
	We use the turnstile algorithm of \cite{kane2010exact}, which gives an estimate $R^t = (1 \pm \eps)\|f^{(t)}\|_p^p$ at a single point $t \in [m]$ with probability $1-\delta$, using $O(\eps^{-2} \log n \log \delta^{-1})$ bits of space. Then for any $\delta_0 \in (0,1)$, we can run this algorithm with failure parameter $\delta = \delta_0/\text{poly}(m)$, and union bound over all steps, to obtain that $R^t = (1 \pm \eps)\|f^{(t)}\|_p^p$  at all time steps $t \in [m]$ with probability $1- \delta_0 $. Thus, this gives a $(\eps,\delta_0)$-strong $F_p$ tracking algorithm using $O(\eps^{-2} \log n \log (n/\delta_0) )$ bits of space. The theorem then follows from applying Lemma \ref{lem:computationpaths}, setting the failure probability to be $n^{-C}$, along with the flip number bound of Lemma \ref{lem:flipbd}. 
\end{proof}

\section{Adversarial Attack Against the AMS Sketch}\label{sec:AMS}

It was shown by \cite{HardtW13} that linear sketches can be vulnerable to adaptive adversarial attacks in the turnstile model, where both insertions and deletions are allowed (see Subsection \ref{subsec:related}). 
In this section, we demonstrate that algorithms based on linear sketching can in some cases be susceptible to attacks even in the \emph{insertion-only} model; Specifically, we show this for the well known Alon-Matias-Szegedy (AMS) sketch \cite{AlonMS96} for estimating the $L_2$ norm of a data stream. 
To this end, we describe an attack fooling the AMS sketch into outputting a value which is not a good approximation of the norm $\|f\|_2^2$ of the frequency vector. Our attack provides an even stronger guarantee: for any $r \geq 1$ and an AMS sketch with $r/\eps^2$ rows, our adversary needs to only create $O(r)$ adaptive stream updates before it can fool the AMS sketch into outputting an incorrect result.

We first recall the AMS sketch for estimating the $L_2$ norm. The AMS sketch generates (implicitly) a random matrix $A \in \R^{t \times n}$ such that the entries $A_{i,j} \sim \{-1,1\}$ are i.i.d.\ Rademacher.\footnote{In fact, the AMS sketch works even if the entries within a row of $A$ are only $4$-wise independent. Here, we show an attack against the AMS sketch if it is allowed to store a fully independent sketch $A$.} The algorithm stores the sketch $A f^{(j)}$ at each time step $j$, and since the sketch is linear it can be updated throughout the stream: $A f^{(j+1)} = A f^{(j)} + A \cdot e_{i_{j+1}} \Delta_{j+1}$ where $(i_{j+1} ,\Delta_{j+1})$ is the $j+1$-st update. The estimate of the sketch at time $j$ is $\frac{1}{t}\|Af^{(j)}\|_2^2$, which is guaranteed to be with good probability a $(1 \pm \eps)$ estimate of $\|f^{(j)}\|_2^2$ in non-adversarial streams if $t = \Theta(\eps^{-2})$. 

We now describe our attack. Let $S$ be a $t \times n$ Alon-Matias-Szegedy sketch. Equivalently, $S_{i,j}$ is i.i.d.\ uniformly distributed in $\{ -t^{-1/2}, t^{-1/2}\}$, and the estimate of AMS is $\|Sf^{(j)}\|_2^2$ at the $j$-th step. The protocol for the adversary is as follows. In the following, we let $e_i \in \R^n$ denote the standard basis vector which is zero everywhere except the $i$-th coordinate, where it has the value $1$.

\begin{algorithm}[!ht]
	\caption{Adversary for AMS sketch} \label{alg:AMS}
    $w \leftarrow C\cdot \sqrt{t} \cdot e_1$\\
    \For{ $i=2,\dots,m$   }{
    old $\leftarrow$ $\|Sw\|_2^2$\\
    $w \leftarrow w +  e_i$ \\
    new $\leftarrow \|Sw\|_2^2$ \\ 
     \If{ $\text{new} - \text{old}< 1$}
     {$w \leftarrow w + e_i$\\
     }
     \ElseIf{ $\text{new} - \text{old} =1$}
     {with probability $1/2$, set $w \leftarrow w + e_i$\\
     }
    }
	
\end{algorithm}
Note that the vector $w$ in Algorithm~\ref{alg:AMS} is always equal to the current frequency vector of  the stream, namely $w =  f^{(j)}$ after the $j$-th update. Note that the Algorithm~\ref{alg:AMS} can be implemented by an adversary who only is given the estimate $\|Sw\|_2^2 = \|Sf^{(j)}\|_2^2$ of the AMS sketch after every step $j$ in the stream. 
To see this, note that the adversary begins by inserting the first item $(i_1,\Delta_1) = (1, C\cdot \sqrt{t} )$ for a sufficiently large constant $C$. Next, for $i=2,\dots,n$, it inserts the item $i \in [n]$ once if doing so increases the estimate of AMS by more than $1$. If the estimate of AMS is increased by less than $1$, it inserts the item $i$ twice (i.e., it inserts an update $(i,2) \in [n] \times \Z$). Lastly, if inserting the item $i \in [n]$ increases the estimate of AMS by \textit{exactly} 1, the adversary chooses to insert $i \in [n]$ once with probability $1/2$, otherwise it inserts $i \in [n]$ twice. 

We now claim that at the end of a stream of $m = O(t)$ updates, with good probability $\|Sf^{(m)}\|_2^2 \notin (1 \pm \eps) \|f^{(m)}\|_2^2$ (note that, at the end of the stream, $w = f^{(m)})$. In fact, we show that regardless of the number of rows $t$ in the AMS sketch, we force the AMS to give a solution that is not even a $2$-approximation. 

\begin{theorem}\label{thm:AMS}
Let $S \in \R^{t \times n}$ be an AMS sketch (i.i.d. Rademacher matrix scaled by $t^{-1/2}$), where $1 \leq t < n/c$ for some constant $c$. Suppose further that the adversary performs the adaptive updates as described in Algorithm \ref{alg:AMS}. Then with probability $9/10$, by the $m$-th stream update for some $m  = O(t)$, the AMS estimate $\|Sf^{(m)}\|_2^2$ of the norm  $\|f^{(m)}\|_2^2$ of the frequency vector $f$ defined by the stream fails to be a $(1 \pm 1/2)$ approximation of the true norm $\|f^{(m)}\|_2^2$. Specifically, we will have $\|Sf^{(m)}\|_2^2  < \frac{1}{2} \|f^{(m)}\|_2^2$. 
\end{theorem}

\begin{proof}
For $j = 2,3\dots$ we say that the $j$-th \textit{step} of Algorithm \ref{alg:AMS} is the step in the for loop where the parameter $i$ is equal to $j$, and we define the first step to just be the state of the stream after line $1$ of Algorithm \ref{alg:AMS}.
Let $w^i$ be the state of the frequency vector at the end of the $i$-th step of the for loop in Algorithm \ref{alg:AMS}, let $y^i = Sw^i$ be the AMS sketch at this step, and let $s_i = \|Sw^i\|_2^2$ be the estimate of AMS at the same point. Note that we have $w^1 = C \cdot \sqrt{t} \cdot e_1$ for a sufficiently large constant $C$, and thus $s_1 = C^2 t$. 
That is, already on the first step of the algorithm we have $\|w^1\|_2^2 = C^2 t$, and moreover since the stream is insertion-only, we always have $\|w^i\|_2^2 \geq C^2 t$. Thus, it suffices to show that with good probability, at some time step $i \geq 2$ we will have $s_i < C^2 t/2$.

First, note that at any step $i=2,3,\dots$, if we add $e_{i+1}$ to the stream once, we have $s_{i+1} = \|y^i + Se_{i+1}\|_2^2 =\sum_{j=1}^t( (y_j^i)^2 + 2y_j^i S_{j,i+1} + 1/t) = s_i + 1 + 2\sum_{j=1}^t y_j^i S_{j,i+1}$. If we add $e_{i+1}$ twice, we have $s_{i+1} = \|y^i + 2Se_{i+1}\|_2^2  = s_i + 4 + 4\sum_{j=1}^t y_j^i S_{j,i+1}$. By definition of the algorithm, we choose to insert $e_{i+1}$ twice if $\|y^i + Se_{i+1}\|_2^2 - s_i = 1 + 2\sum_{j=1}^t y_j^i S_{j,i+1} < 1$, or more compactly whenever $\sum_{j=1}^t y_j^i S_{j,i+1} < 0$. If $\sum_{j=1}^t y_j^i S_{j,i+1} > 0$, we insert $e_{i+1}$ only once. Finally, if $\sum_{j=1}^t \allowbreak y_j^i S_{j,i+1} \allowbreak= 0$, we flip an unbiased coin, and choose to insert $e_{i+1}$ either once or twice with equal probability $1/2$. Now observe that the random variable $\sum_{j=1}^t y_j^i S_{j,i+1}$ is symmetric, since for any fixed $y^i$ the $S_{j,i+1}$'s are symmetric and independent. Thus, we have that
\begin{equation}
    \begin{split}
      \bex{\Big|\sum_{j=1}^t y_j^i S_{j,i+1} \Big| }   &= \bex{\sum_{j=1}^t y_j^i S_{j,i+1}  \; | \; Se_{i+1} \text{ inserted once} }\\
        &=  -\bex{\sum_{j=1}^t y_j^i S_{j,i+1}  \; | \; Se_{i+1} \text{ inserted twice} }. \\
    \end{split}
\end{equation}

Now recall that the vector $S_{*,i+1}$ given by the $(i+1)$-st column of $S$ is just an i.i.d. Rademacher vector scaled by $1/\sqrt{t}$. Thus, by Khintchine's inequality \cite{haagerup1981best}, we have that  $\ex{|\sum_{j=1}^t y_j^i S_{j,i+1} | } = \frac{1}{\sqrt{t}} \cdot \alpha \cdot \|y^i\|_2 =  \alpha \sqrt{s_i}/\sqrt{t}$ for some absolute constant $\alpha > 0$ (in fact, $\alpha \geq 1/\sqrt{2}$ suffices by Theorem 1.1 of \cite{haagerup1981best}). 
Putting these pieces together, the expectation of the estimate of AMS is then as follows:
\begin{equation}
    \begin{split}\label{eqn:amsexp}
      \E[s_{i+1}] &= \frac{1}{2} (s_i + 1 + 2\alpha \frac{\sqrt{s_i}}{\sqrt{t}}) + \frac{1}{2}(s_i + 4  - 4\alpha \frac{\sqrt{s_i}}{\sqrt{t}}) \\
&= s_i + 5/2- \alpha \sqrt{s_i/t} \\ 
&\leq s_i + 5/2-  \sqrt{s_i/2t} .
    \end{split}
\end{equation}
Where again the last line holds using the fact that $\alpha \geq 1/\sqrt{2}$. 
Thus $\ex{s_{i+1}} = \ex{s_i} + 5/2 - \ex{  \sqrt{s_i/2t}}$. 
First, suppose there exists some $i \leq  C^2 t + 2$ such that $\ex{  \sqrt{s_i}} < C\sqrt{t/200}$. This implies by definition that
$  \sum_j \sqrt{j} \cdot \pr{s_i = j} <     C\sqrt{t/200}$, thus 
\begin{equation}
   \sqrt{C^2t/2} \cdot \pr{s_i \geq C^2 t/2} \leq \sum_{j \geq C^2 t/2} \sqrt{j} \cdot \pr{s_i = j} 
     <  \sqrt{C^2 t/200}
\end{equation}
Which implies that $\pr{s_i \geq C^2 t/2} \leq 1/10$. Thus, at step $i$, we have $\pr{s_i < C^2 t/2} > 9/10$, and thus by time step $i$ we have fooled the AMS sketch with probability at least $9/10$. Thus, we can assume that for all $i=2,3,\dots,(C^2 t + 2)$ we have $\ex{  \sqrt{s_i}}  \geq C\sqrt{t/200}$. Setting $C>200$, we have that $\ex{s_{i+1}} < \ex{s_i} -1$ for all steps $i=2,3,\dots,(C^2 t + 2)$   However, since $s_1 = C^2 t$, this implies that $\ex{s_{C^2 t + 2}} < -1$, which is impossible since $s_j$ is always the value of a norm. This is a contradiction, which implies that such an $i$ with $i \leq  C^2 t + 2$ and $\pr{s_i \geq C^2 t/2} \leq 1/10$ must exist, demonstrating that we fool the AMS sketch by this step with probability $9/10$, which
completes the proof.
\end{proof}

\section{Optimal Distinct Elements via Cryptographic Assumptions}
\label{sec:crypto}
Estimating the number of distinct elements ($F_0$-estimation) in a data stream is a fundamental problem in databases, network traffic monitoring, query optimization, data mining, and more. After a long line of work, \cite{woodruff2004optimal,kane2010optimal} settled space (and time) complexities of $F_0$-estimation by giving an algorithm using $O(\varepsilon^{-2} + 	\log n)$ bits of space (with constant worst-case update time). The tracking version of this algorithm (where it outputs a correct estimate at each time step) takes memory $O(\varepsilon^{-2}(\log \eps^{-1} + \log \log n) + 	\log n)$ bits and is also optimal \cite{blasiok2018optimal}. 

However, these results only hold in the (standard) static setting. We show that using cryptographic tools (pseudorandom functions), we can transform this algorithm, using the same amount of memory to be robust in the adversarial setting as well, where the adversary is assumed to be {\em computationally bounded} (as opposed to our other results which have no assumptions on the adversary whatsoever).

The transformation actually works for a large class of streaming algorithms. Namely, any algorithm such that when given an element that appeared before, does not change its state at all (with probability 1). Since the $F_0$ tracking algorithm of \cite{blasiok2018optimal} has this property, we can black-box apply our results to this algorithm. 

First, we show how this transformation works assuming the existence of a truly random function, where the streaming algorithm has access to the function without needing to store it explicitly (the memory is free). This is known as the random oracle model. 
The model is appealing since we have different heuristic functions (e.g., SHA-256) that behave, as far as we can tell in practice, like random functions. Moreover, there is no memory cost when using them in an implementation, which is very appealing from a practical perspective.
Nevertheless, we discuss how to implement such a function with cryptographic tools (e.g., pseudorandom functions) while storing only a small secret key in the memory. 

\begin{theorem}[Distinct Elements by Cryptographic Assumptions]\label{thm:distinct_elements_crypto}
	In the random oracle model, there is an $F_0$-estimation (tracking) streaming algorithm in the adversarial setting, that for an approximation parameter $\varepsilon$ uses $O(\varepsilon^{-2}(\log 1/\varepsilon + \log \log n) + 	\log n)$ bits of memory, and succeeds with probability $3/4$.
	
	Moreover, given an exponentially secure pseudorandom function, and assuming the adversary has bounded running time of $n^c$, where  $c$ is fixed, the random oracle can be replaced with a concrete function and the total memory is $O(\varepsilon^{-2}(\log 1/\varepsilon + \log \log n) + c\log n)$.
\end{theorem}

\begin{proof}
	For simplicity, in the following proof, we assume that we have a random permutation. We note that the proof with a random function is exactly the same conditioned on not having any collisions. If the random function maps the universe to a large enough domain (say of size at least $m^2$) then there will be no collisions with high probability. Thus, it suffices to consider permutations.
	The solution is inspired by the work of \cite{NaorY15} (which had a similar adaptive issue in the context of Bloom filters).
	Let $\Pi$ be a truly random permutation, and let $S$ be a tracking steaming algorithm with parameter $\varepsilon$. Let $L(\varepsilon,n)$ be the memory consumption of the algorithm. We construct an algorithm $S'$ that works in the adversarial setting as follows. Upon receiving an element $x$ the algorithm $S'$ computes $x'=\Pi(x)$ and feeds it to $S$. The output of $S'$ is exactly the output of $S$. Notice that applying $\Pi$ to the stream does not change the number of distinct elements.
	
	We sketch the proof. Assume towards a contradiction that there is adaptive adversary $A'$ for $S'$.
	Consider the adversary $A'$ at some point in time $t$, where the stream is currently $x_1,\ldots,x_t$.
	It has two options: (i) it can choose an element $x_i$, where $i \in [t]$ that appeared before, or (ii) it could choose a new element $x^* \notin \{x_1,\ldots,x_i\}$. Since the state of $S'$ does not change when receiving duplicate items, and also does not change the number of distinct elements, option (i) has no effect on the success probability of $A'$. Thus, in order to gain a chance of winning, $A'$ must submit a new query. Thus, we can assume without loss of generality that $A'$ submits only distinct elements.
	
	For such an adversary $A'$ let $D_t$ be the distribution over states of $S'$ at time $t$. Let $D'_t$ be the distribution over states of $S'$ for the fixed sequence $1,2,\ldots,t$. We claim that $D_t \equiv D'_t$ (identical distributions) for every $t \in [n]$. We show this by induction. The first query is non-adaptive, denote it by $x_1$. Then, since $\Pi$ is a random permutation, we get that $\Pi(1) \equiv\Pi(x_1)$ which is what is fed to $S$. Thus, the two distribution are identical. Assume it holds for $t-1$. Consider the next query of the adversary (recall that we assumed that this is a new query). Then, for any $x_t$ (that has not been previously queried by $\Pi$) the distribution of $\Pi(x_t) \equiv \Pi(t)$, and therefore we get that $D_t \equiv D'_t$.
	
	Given the claim above, we get that $A'$ is equivalent to a static adversary $A$ that outputs $1,2,\ldots,k$ for some $k \in [n]$. However, the choice of $k$ might be adaptive. We need to show that $S'$ works for all $k$ simultaneously. Here we use the fact that $S$ was a tracking algorithm (and thus also $S'$), which means that $S'$ succeeds on every time step. Thus, for the stream $1,2,\ldots,m$, the algorithm $S'$ succeeds at timestamp $k$, which consists of $k$ distinct elements. Thus, if there exists an adaptive choice of $k$ that would make $S'$ fail, then there would exist a point in time, $k$, such that $S'$ fails at $1,\ldots,k$. Since $S$ is tracking, such a point does not exist (w.h.p.).
	
	For the second part of the theorem, we note that we can implement the random function using an exponentially secure pseudorandom function (see \cite{Goldreich05} for the precise definition and discussion). For a key $K$ of size $\lambda$, the pesudorandom function $F_K(\cdot)$ looks random to an adversary that has oracle access to $F_K(\cdot)$ and runs in time at most $2^{\gamma\lambda}$ for some constant $\gamma>0$. Let $A$ be an adversary that runs in time at most $n^c$. Then, we set $O(\lambda=1/\gamma \cdot c \cdot \log n)$ and get that $A$ cannot distinguish between $F_K(\cdot)$ and the truly random function except when a negligible probability event occurs (i.e., the effect on $\delta$ is negligible and hidden in constants). Indeed, if $A$ would be able to succeed against $S'$ when using the oracle $F_K(\cdot)$, but, as we saw, it does not succeed when using a truly random function, then $A'$ could be used to break the security of the pseudorandom function.
		
	To complete the proof, we note that the only property of $A$ we needed was that when given an element in the stream that has appeared before, $A$ does not change its state at all. This property holds for many $F_0$ estimation algorithms, such as the one-shot $F_0$ algorithm of \cite{kane2010optimal}, and the $F_0$ tracking algorithm of \cite{blasiok2018optimal}. Thus we can simply use the $F_0$ tracking algorithm of \cite{blasiok2018optimal}, which results in the space complexity as stated in the theorem.
\end{proof}

\begin{remark}
		There are many different ways to implement such a pseudorandom function with exponential security and concrete efficiency. First, one could use heuristic (and extremely fast) functions such as AES or SHA256 (see also \cite{NaorY15} for a discussion on fast implementations of AES in the context of hash functions). Next, one can assume that the discrete logarithm problem (see \cite{mccurley1990discrete} for the precise definition) over a group of size $q$ is exponentially hard. Indeed, the best-known algorithm for the problem runs in time $O(\sqrt{q})$. Setting $q \ge 2^{\lambda}$ gets us the desired property for $\gamma=1/2$. 
\end{remark}

\section*{Acknowledgments}
The authors wish to thank Arnold Filtser for invaluable feedback, and the anonymous reviewers for many helpful suggestions. This work was done in part in the Simons Institute for the Theory of Computing.

\bibliographystyle{alpha}
\bibliography{ref-v2}


\end{document}